\documentclass[envcountsame]{llncs}

\usepackage[utf8]{inputenc} 
\usepackage[T1]{fontenc}    
\usepackage[hidelinks]{hyperref}       
\usepackage{url}            
\usepackage{booktabs}       
\usepackage{amsfonts}       

\usepackage{amsthm}
\usepackage{amsmath}
\usepackage{amssymb}
\usepackage{nicefrac}       
\usepackage{microtype}      
\usepackage[table]{xcolor}
\usepackage{algorithm}
\usepackage[noend]{algorithmic}
\usepackage{thm-restate}
\usepackage{xspace}
\usepackage{mathtools}
\usepackage{accents}
\usepackage{bm,bbm}
\usepackage{relsize}
\usepackage{cleveref}
\usepackage{wrapfig}
\usepackage{fullpage}
\usepackage{cite}
\usepackage{paralist}
\usepackage{multirow}

\newcommand{\Rade}{\mathcal{R}}

\newcommand{\ProbDist}{\mathcal{D}}
\newcommand{\HC}{\mathcal{H}}
\newcommand{\transpose}{\mathrm{T}}
\newcommand{\rank}{\operatorname{rank}}
\newcommand{\wv}{\vec{w}}
\newcommand{\loss}{\ell}
\newcommand{\Sorted}{\mathit{Sorted}}
\newcommand{\defeq}{\doteq}

\newcommand{\distributed}{\thicksim}
\newcommand{\Prob}{\mathbb{P}}
\newcommand{\TVD}{\operatorname{TVD}}

\DeclareMathOperator*{\Expect}{\mathbb{E}}

\newcommand{\cyrus}[1]{{}}
\newcommand{\justin}[1]{{}}

\renewcommand{\problemname}{\textsc{RAU}\xspace}
\newcommand{\algoname}{\textsc{RRA}\xspace}
\newcommand{\LP}{\textsc{LP}\xspace}

\renewcommand{\Box}{\mathsmaller{\square}}
\newcommand{\COI}{\textsc{CoI}}
\newcommand{\AMM}{A^{\mathrm{MM}}}
\newcommand{\AMMC}{\tilde{A}^{\mathrm{MM}}}
\newcommand{\ssep}{\,\mid\,}

\newcommand{\R}{\mathbb{R}}

\renewcommand{\cal}[1]{\mathcal{#1}}

\usepackage{eqparbox}

\DeclareMathOperator*{\argmax}{arg\,max}
\DeclareMathOperator*{\argmin}{arg\,min}
\DeclareMathOperator*{\arginf}{arg\,inf}

\DeclareMathOperator*{\USW}{USW}

\DeclareMathOperator*{\ROUND}{\textsc{Round}}

\usepackage{mathtools}

\DeclarePairedDelimiter\abs{\lvert}{\rvert}%
\DeclarePairedDelimiter\norm{\lVert}{\rVert}

\newcommand{\SBox}{\cal S_{\mathsmaller{\square}}}

\usepackage{upgreek}

\providecommand{\LandauO}{\bm{\mathrm{O}}} 

\providecommand{\LandauOmega}{\bm{\Upomega}}

\usepackage[T1]{fontenc}

\usepackage{setspace}
\setdisplayskipstretch{0.5}

\usepackage{enumitem}
\setlist[enumerate]{nosep,itemsep=0pt}
\setlist{nosep,wide,labelwidth=0pt,labelindent=0pt,labelsep=3pt,topsep=-2pt}

\usepackage[font=small,labelfont=bf]{caption} 

\let\oldCref\Cref
\renewcommand{\cref}{\Cref} 

\renewcommand{\USW}{\mathrm{W}}

\begin{document}

\title{Into the Unknown: \\ Assigning Reviewers to Papers  with Uncertain Affinities}

\author{Cyrus Cousins\orcidID{0000-0002-1691-0282}, 
Justin Payan\thanks{Corresponding Author. All authors contributed equally to this work.}\orcidID{0000-0001-7601-3500}, \\ and 
Yair Zick\orcidID{0000-0002-0635-6230} \\
}
\authorrunning{C. Cousins et al.}
%
\institute{University of Massachusetts Amherst, Amherst MA 01002, USA \\
\email{\{cbcousins, jpayan, yzick\}@umass.edu}}

\maketitle

\begin{abstract}
A successful peer review process requires that qualified and interested reviewers are assigned to each paper. 
Most automated reviewer assignment approaches estimate a real-valued \emph{affinity score} for each paper-reviewer pair that acts as a proxy for the quality of the match,
and then assign reviewers to maximize the sum of affinity scores. 
Most affinity score
estimation
methods are inherently noisy: reviewers can only bid on a small number of papers, and textual similarity models 
and subject-area matching
are inherently noisy estimators.
Current paper assignment systems are not designed to rigorously handle noise in the peer-review matching market.   
In this work, we assume paper-reviewer affinity scores are located in or near a high-probability region called an \emph{uncertainty set}. 
We
maximize the worst-case sum of scores for a reviewer assignment over the uncertainty set. 
We demonstrate how to robustly maximize the sum of scores across
various
classes of uncertainty sets, avoiding potentially serious mistakes in assignment. 
Our general approach can be used to integrate a large variety of paper-reviewer affinity models into reviewer assignment, opening the door to a much more robust peer review process.

\end{abstract}
\section{Introduction}
\label{sec:intro}
Peer review is a fundamental institution for evaluating scientific knowledge. 
Over the 20th century, the scientific profession has grown significantly, and the institution of peer review has struggled with the increased scale. 
Modern computer science conferences receive thousands of submissions, matched to committees of similar size.  
Due to sheer scale, program chairs rely on automated reviewer assignment platforms, such as Microsoft CMT, OpenReview, and EasyChair, that utilize complex matching algorithms.
Reviewing platforms generally implement a two-stage process to largely automate reviewer assignments
\cite{leyton2022matching,charlin2013toronto,openreview_workflow}. 
First, the system estimates the ``fit'' between each paper-reviewer pair, called the \emph{paper-reviewer affinity score}. 
Next, the system assigns reviewers via constrained optimization, maximizing a function of the computed affinity scores (usually the sum of scores, or \emph{utilitarian welfare} \cite{charlin2013toronto}). 
Prior work identifies the reviewer assignment process as an important target for improving the overall quality of peer review in computer science \cite{rogers2020can,shah2022challenges}. 
Efforts are underway to address shortcomings in reviewer assignment \cite{stelmakh2019peerreview4all, kobren2019paper, payan2022order, jecmen2020random, jecmen2022dataset, shah2019principled, acl_score, leyton2022matching, meir2021market, rozencweig2023mitigating}, but none systematically addresses the fundamental issue of \emph{uncertainty} in reviewer assignment.

Uncertainty in affinity score computation is a major source of error in assignment \cite{leyton2022matching}. 
When we assign a reviewer to a paper, we are interested in ensuring the \emph{quality of the future review}, which is fundamentally noisy. 
Because of this unpredictability, conferences typically construct affinity scores that reflect reviewer expertise and interest via four main sources of information. These sources include
\begin{inparaenum}[(a)]
\item \emph{subject-area matching} (SAM) scores or keyword-based matching, where reviewer-provided areas of expertise are compared against keywords submitted by paper authors,
\item \emph{textual similarity scores}, often implemented by the well-known Toronto Paper Matching System (TPMS) \cite{charlin2013toronto} or ACL scores \cite{acl_score},
\item \emph{bidding}, where reviewers express their explicit ability and desire to review papers, and finally
\item \emph{recommendations}, through which program committee members may suggest reviewers for papers.
\end{inparaenum}
The overall affinity scores are typically computed as a linear combination of these four scores.\footnote{CMT implements their affinity scores this way, which can be seen from \url{https://cmt3.research.microsoft.com/docs/help/chair/auto-assign-reviewers.html}, as does OpenReview (source: personal correspondence).}
Recent conferences such as AAAI 2021 took a similar approach, linearly combining TPMS scores, ACL scores, and SAM scores, and raising the sum to some power based on the reviewer bids \cite{leyton2022matching}. 

Each of these common affinity score components can be missing or inaccurate. 
State-of-the-art document similarity measures disagree with expert judgments up to $43\%$ of the time \cite{stelmakh2023gold}, and nearly 40\% of TPMS scores were completely missing in AAAI 2021 \cite{leyton2022matching}.\footnote{Although the AAAI 2021 organizers do not explain
why so many TPMS scores are missing,
missing scores occur for several reasons, including 
reviewers opting out of the system or providing insufficient or empty publication records.%
} 
Between $5\%$ and $15\%$ of papers in major AI conferences receive fewer than three positive bids, but 
there is
evidence that many missing bids would be positive if collected \cite{pmlr-v124-fiez20a, meir2021market, rozencweig2023mitigating}. 
Although no systematic study has been performed on keyword-based similarity scores, keyword matching accuracy depends on authors and reviewers using consistent terminology, and subtleties are invariably lost in the process. 
Even reviewers directly suggested by knowledgeable editors or the paper authors have been shown to perform surprisingly poorly on average, as measured by third-party annotators via the Review Quality Index \cite{schroter2006differences, van1999development}, showing that recommendations can be noisy as well.

To our knowledge, every reviewer assignment system still relies on affinity score estimates, but does not directly account for the fact that these scores are
noisy estimates of assignment quality. 
Our work
takes the first step towards addressing this fundamental gap. We investigate a generalized notion of affinity score, where organizers can implement affinity using any measure of fit between reviewers and papers. These measures may or may not be fully observable; for example, organizers may decide to estimate unknown bids as part of affinity computation. Our approach also enables even more advanced affinity measures, such as predictors of reviewer performance based on historical data.
%


\subsection{Our Contributions}

To properly account for uncertainty, we first construct a region called an \emph{uncertainty set} which is close to the true affinities with high probability, then maximize the worst-case welfare over the uncertainty set. 
Uncertainty sets are a very general construct that allows conference organizers to introduce their own uncertainty models using available data and reasonable assumptions. 
Uncertainty sets generalize probability distributions --- while it is possible to construct an uncertainty set from a probability distribution, non-Bayesian models will frequently \emph{not} specify full probability distributions.
In these cases, worst-case guarantees over an uncertainty set are quite natural. 
We call the problem of maximizing the worst-case welfare over an uncertainty set Reviewer Assignment under Uncertainty, or \problemname{}.

We provide numerous examples of uncertainty sets throughout the paper, starting with axis-aligned, hyperrectangular uncertainty sets in \Cref{subsec:box} and spherical and ellipsoidal uncertainty sets in \Cref{subsec:ellipsoid}. 
\oldCref{thm:inductiveellipsoidal-tail-bounds,thm:transductiveellipsoidal-tail-bounds} offer detailed end-to-end examples of how conference organizers
can
construct 
ellipsoidal uncertainty sets 
using
bounds on the \emph{square error} of an \emph{affinity score estimator} from historic data or sampled bids (a 
frequentist approach that precludes optimizing for expected welfare). 
We also
present a \emph{calculus of uncertainty sets}, 
enabling
construction of complex and highly informative uncertainty sets from 
simple components
(\Cref{subsec:compositional}). Our results are agnostic to the affinity model;
organizers can define
affinity scores arbitrarily,
so long as
they
can be estimated for all paper-reviewer pairs, and 
sampled for some pairs.


We show that
\problemname{} is NP-hard over convex uncertainty sets (\Cref{thm:hardness}), and present an approximation algorithm called Robust Reviewer Assignment (\algoname{}), which applies to any convex uncertainty set where worst-case welfare can be
efficiently computed 
(\Cref{sec:approx_algo}). 
\algoname{} applies \emph{randomized rounding} methods 
to a \emph{convex relaxation} of the discrete RAU problem, and we analyze both the \emph{optimization error} 
due to convex programming methods and randomized rounding, and the \emph{maximin error} due to operating with an uncertainty set, rather than known affinity scores (\Cref{subsec:gaps}).
We  give bounds on the true welfare relative to the maximin welfare solution of \problemname{} (\Cref{prop:sandwich}), and explore the integrality gap of \problemname{} (\Cref{prop:integ_gap_lb}).


We
empirically demonstrate
the robustness of our approach relative to commonly-used baselines on publicly available data from five recent iterations of ICLR (\Cref{subsec:comparisontobaseline}). 
In addition, we explore synthetic settings where \algoname{} avoids negative consequences faced by the most commonly used baseline (\Cref{subsec:worstcase}). 
We hope this work draws attention to the ad-hoc nature of affinity scores, spurring improvements to their computation and further study of their robustness. All proofs in this paper are contained in the appendix.

\subsection{Related Work}


Current automated peer review systems compute pointwise, ad-hoc affinity scores and maximize the total sum~\cite{charlin2013toronto, leyton2022matching, payan2022order}, the minimum value for papers~\cite{kobren2019paper, stelmakh2019peerreview4all}, or the minimum value for groups of papers \cite{aziz2023group}.

Two recent studies start from our 
premise that 
commonly-used affinity scores may not be as
accurate
as they seem. Data is now available that directly compares elements of affinity scores to expert judgments of reviewer fit, and the authors of this dataset show that existing similarity score computation methods make many errors \cite{stelmakh2023gold}. A 
recent study leverages randomness in assignment algorithms to directly judge assignment decisions, showing that higher weight should be placed on text similarity metrics over bids \cite{saveski2023counterfactual}. These results encourage further improvements to affinity score computation  (especially text similarity), but also justify smarter utilization of the noisy sources of information
that are
available.



Other works 
use modern NLP techniques to improve document-based similarity scores \cite{acl_score},  encourage reviewers to bid on underbid papers \cite{pmlr-v124-fiez20a,meir2021market,rozencweig2023mitigating}, or disincentivize strategic bidding behavior \cite{jecmen2020random,jecmen2022tradeoffs,jecmen2022dataset}. 
Although these approaches reduce uncertainty, they do not directly 
treat uncertainty in affinity scores. 
%


Our robust optimization algorithm is based on an iterative supergradient-ascent approach; similar techniques have been applied to supervised learning with unknown labels \cite{mazzetto2021adversarial} and fair 
learning
with unknown group identities \cite{dong2022decentering}.

\section{Reviewer Assignment under Uncertainty}
\label{sec:problem}
Assume we have a set $P$ of $n$ papers submitted to a peer-review venue,\footnote{A peer review venue is any entity which assigns reviewers to papers for the purposes of peer review. The prototypical venue is a peer-reviewed conference,
but our approach applies to similar venues such as ACL Rolling Review \cite{acl_rolling_review}.}
and a set $R$ of $m$ reviewers. 
The key input to the reviewer assignment problem is a paper-reviewer \emph{affinity score matrix} $S^\ast \in [0,1]^{n\times m}$ (we will also refer to $[0,1]^{n\times m}$ as the \emph{unit hypercube}), where $S^\ast_{p, r}\in [0,1]$ is the affinity  of paper $p \in P$ to reviewer $r \in R$. 
This matrix $S^\ast$ encodes the \emph{true} affinities, or the value provided to the venue by assigning each reviewer to each paper. At this point, it is natural to ask,
``How can one know the value a paper-reviewer assignment will provide ahead of time?'' and ``What do we mean by value provided?''
These questions directly motivate our work; reviewer assignment is challenging because affinity depends on poorly-defined preferences over uncertain future outcomes.

Despite this fundamental challenge, all prior work assumes direct access to the true affinities \cite{kobren2019paper, stelmakh2019peerreview4all, charlin2013toronto, payan2022order,leyton2022matching,aziz2023group}. We
relax this assumption, instead assuming only
\emph{partial knowledge} of $S^\ast$, as represented by an \emph{uncertainty set} $\mathcal{S} \subseteq [0, 1]^{n \times m}$ that contains a point near $S^\ast$ with high probability. 

\begin{definition}[$(\delta, \gamma)$ Uncertainty Set]
\label{defn:dg_uncert}
Suppose $S^\ast \in [0, 1]^{n \times m}$ is the
true affinity score matrix.
A $(\delta, \gamma)$ uncertainty set 
obeys $\Prob \bigl( \smash{\inf\limits_{S \in \cal{S}}} \norm{S - S^\ast}_1 > \gamma \bigr) < \delta$, i.e., it probably contains
some $S$
that is 
$\gamma$-close to $S^\ast$. 
\end{definition}

Once we compute a $(\delta, \gamma)$ uncertainty set $\cal S$, our goal is to find an \emph{assignment matrix} $A$ that assigns reviewers to papers with good worst-case guarantees on a score function defined by $A$ and $S^\ast$, while satisfying some hard constraints.  Assignments are deterministic, i.e. lie on vertices of the unit hypercube, or $\cal A_0 \defeq \{0, 1\}^{n \times m}$.
We compute an assignment of reviewers to papers $A \in \cal A_0$, where $A_{p,r} = 1$ if and only if $r$ is assigned to review the paper $p$.  
Hard constraints usually include the number of reviewers required per paper, upper bounds on reviewer loads, and conflicts of interest.
Conflicts of interest consist of a set $\mathcal{C} \subseteq (P \times R)$, where $(p, r) \in \mathcal{C}$ implies that $r$ cannot be assigned to $p$. 
$\mathcal{A}_{\COI} \subseteq \cal A_0$ denotes the set of all assignments respecting conflicts in $\cal C$. If $A \in \mathcal{A}_{\COI}$ and $(p,r) \in \cal C$, then $A_{p,r}=0$.

Suppose that each paper $p$ requires exactly $k_p$ reviewers and each reviewer $r$ must be assigned no more than $u_r$ papers. Then 
\[
\mathcal{A}_P \doteq \Big\{ \! A \! \in \! \mathcal{A} _{0}\! \Bigm| \forall p \! \in \! P\!: \raisebox{0.5ex}{\scalebox{0.95}{\ensuremath{\displaystyle\mathsmaller{\sum\limits_{r \in R}}}}}A_{p,r} = k_p\Big\} 
\ \& \ 
\mathcal{A}_R  \doteq \Big\{ \! A \! \in \! \mathcal{A}_{0} \!  \Bigm| \forall r \! \in \! R \!: \raisebox{0.5ex}{\scalebox{0.95}{\ensuremath{\displaystyle\mathsmaller{\sum\limits_{p \in P}}}}} A_{p,r} \leq u_r \Big\}
\]
define the \emph{paper coverage requirements} and \emph{reviewer load bounds}, respectively.
We occasionally refer to the \emph{total review load} $K \doteq \sum_{p \in P} k_p$.
%
%
%
%
%
Taken together, the hard constraints give us a set of permissible assignments $\mathcal{A} \doteq \mathcal{A}_{\COI} \cap \mathcal{A}_P \cap \mathcal{A}_R$.




We aim to compute assignments $A$ that maximize a score-based objective function $\USW(A, S)$, while meeting all hard constraints. Throughout the paper we assume $\USW$ is utilitarian social welfare
$\USW(A, S) \doteq \frac{1}{n}\sum_{p \in P}\sum_{r \in R} A_{p,r}S_{p,r}$.


In summary, the problem of Reviewer Assignment under Uncertainty (\problemname{}) takes as input a set of papers $P$, reviewers $R$, assignment constraints $\cal A$, and an uncertainty set $\cal S$. 
Although ideally we would compute $A^\ast \defeq \argmax_{A \in \cal A} \USW(A, S^\ast)$, we cannot since we do not know $S^\ast$. Therefore, our goal is to find or approximate
\begin{align}
\label{def:AMM}
    \AMM \doteq \argmax_{A \in \cal A }
    \inf_{S \in \cal S}
    \USW(A, S) \enspace.
\end{align}
%
Here $\AMM$ maximizes welfare for 
adversarial (worst-case)
affinity scores
 over the uncertainty set, which ensures robustness
 with high probability. 
Robust guarantees
prevent
catastrophic failures in exchange for some average case penalty. 

We conclude this section with a simple observation relating the worst-case welfare over $\cal S$ for any assignment $A$ to the true welfare of $A$ with $S^\ast$. 
\begin{restatable}[Relating True and Worst-Case Welfare over $\cal S$]{proposition}{propsandwich}
\label{prop:sandwich}
Suppose $\cal S$ is a $(\delta, \gamma)$ uncertainty set with $\norm{S-S'}_1 \leq L$ for all $S, S' \in \cal S$, and the true affinity score matrix is labeled $S^\ast$. Consider any assignment $\begin{aligned} A \in \cal A\end{aligned}$\cyrus{Doesn't this hold for any assignment $A$?}. Then with probability at least $1-\delta$,  \begin{align*} \USW(A, S^\ast) - \mathsmaller{\frac{L + \gamma}{n}} \leq \smash{\inf_{S \in \cal S}} \USW(A, S) \leq \USW(A, S^\ast)+
\mathsmaller{\frac{\gamma}{n}} \enspace.\end{align*}%
\cyrus{Could also give a sharper radius formulation}%
\cyrus{Diameter computation can be very hard (computationally), give form in terms of $\max\max$? I.e.,
\[
\USW(A^\ast, S^\ast) - \gamma \leq \smash{\max_{S \in \cal S}} \USW(A^\ast, S) \leq \inf_{S \in \cal S} \USW(A^\ast, S) + L \leq \USW(A^\ast, S^\ast)+L+\gamma
\]}%
\cyrus{Do we want to compare true welfare of robust solution to true welfare of true solution here? And consider rounding and other details later?}%
\end{restatable}
%
\oldCref{prop:sandwich} implies that if we aim for low $\cal L_1$ diameter uncertainty sets $\cal S$, we can
approximately optimize true welfare using the robust objective \eqref{def:AMM}.
A similar bound holds when we compare $\USW(\AMM, S^\ast)$ to $\USW(A^\ast, S^\ast)$, i.e., to the welfare of the unknown, true optimal assignment $A^\ast$  (\Cref{thm:gaps}). However, \Cref{prop:sandwich} states that the maximin objective itself is close to the true welfare under any assignment $A$, whereas \Cref{thm:gaps} shows that the specific assignment $\AMM$ has low regret against  $A^\ast$.




\section{Uncertainty Models for Affinity Scores}
\label{sec:uncertainty_set_types}

We start by analyzing simple uncertainty sets, namely the case where $S$ is known, as well as hyperrectangular, spherical, and ellipsoidal uncertainty sets. We conclude the section by showing a compositionality rule that allows venue organizers to combine multiple uncertainty sets into a single \problemname{} problem.

The case where $S$ is known, $\mathcal{S} = \{S\}$, is quite straightforward. This corresponds to most current conference management system implementations \cite{leyton2022matching,charlin2013toronto}. When $\mathcal{A} = \mathcal{A}_{\COI} \cap \mathcal{A}_P \cap \mathcal{A}_R$, the binary integer program $\argmax_{A \in \mathcal{A}}\USW(A, S)$
is known to be polynomial-time solvable, as it is a linear program with totally unimodular constraints. Intuitively, these constraints introduce cuts to the assignment hypercube that never produce new vertices, thus all vertices of the resulting polytope occur on the binary integer lattice. Efficient implementations of this linear program can be found in multiple publicly available sources \cite{taylor2008optimal,kobren2019paper,stelmakh2019peerreview4all,payan2022order}.
Let us now consider hyperrectangular, spherical, and ellipsoidal uncertainty sets. 


\subsection{Hyperrectangular Uncertainty Sets from Confidence Intervals}
\label{subsec:box}

Many simple and intuitive models for an uncertainty set take the form of axis-aligned hyperrectangles. 
A na{\"i}ve uncertainty set might estimate confidence intervals for each $S_{p,r}$ independently and use a union bound to give a high-probability region for the affinity scores. 
Venue organizers might also make assumptions about intervals bounding affinity scores with certainty, taking the intersection of multiple such interval bounds. 
For example, they might start with the global constraints of the unit hypercube. 
Lower and upper bounds can then be given for pairs based on whether they receive certain bids, whether the program committee recommends the assignment, or whether a threshold on document similarity score is met. 
This model is ad-hoc and simple, but may be suitable in practice.
Furthermore, if we assume that with probability at least $1 - \delta$, only a small constant fraction $\gamma$ of these bounds can be violated, we can establish a $(\delta, \gamma)$ confidence interval under more realistic assumptions.

If we take all the lower bounds on paper-reviewer scores, we obtain a lower bound affinity score matrix $\underline{S}$. Similarly, taking all the maximal possible values for paper-reviewer scores yields an upper bound affinity score matrix $\bar{S}$. Our uncertainty set is thus $\mathcal{S}_{\Box} \doteq \{X \in \R^{n \times m} \mid \underline{S}_{p,r} \leq X_{p,r} \leq \bar{S}_{p,r} \ \forall p,r\}$.
Our first result is that \problemname{} can be solved in polynomial time for axis-aligned, hyperrectangular uncertainty sets. 
%
\begin{restatable}[
RAU
under Hyperrectangular Uncertainty]{theorem}{thmbox}\label{thm:box}
    When the  uncertainty set is 
    an axis-aligned hyperrectangular
    region
    $\cal S_{\Box}$, then
    \[    \argmax_{A \in \mathcal{A}} \inf_{S \in \SBox} \USW(A, S) = \argmax_{A \in \mathcal{A}} \USW(A, \underline{S}) \enspace.
    \]
    Thus RAU 
    maximizes $\USW$ for $\underbar{S}$, 
    which requires polynomial time via LP reduction. 
\end{restatable}
%
%
Axis-aligned, hyperrectangular uncertainty sets correspond to the case where uncertainty is bounded independently across paper-reviewer scores, hence their relative simplicity. Although hyperrectangular uncertainty sets are easy to work with, they are unnecessarily pessimistic, since it is very unlikely that all affinities take extreme values at once (i.e., $\underbar{S}$ is actually a very unlikely outcome).
We can improve our estimates using uncertainty set models that account for the low probability of many simultaneous extreme values.


\subsection{Ellipsoidal Uncertainty Sets with $\cal L_2$ Error Guarantees}
\label{subsec:ellipsoid}

Many standard models directly bound the $\cal L_1$ or  $\cal L_2$ error of their predictions, which implies uncertainty sets that are more optimistic than hyperrectangular $\cal S$ (and hence have tighter guarantees for \Cref{prop:sandwich}).

We first analyze the case of symmetric uncertainty sets with $\cal L_2$ error guarantees.
For example, we might solicit bids uniformly at random  and then predict unsampled bids using collaborative filtering with $\cal L_2$ error guarantees \cite{lee2010practical,cai2016matrix,fang2018max}. $\mathcal{S}$ could then be constructed as a linear combination of values known with certainty (document-based similarity scores and keyword-based matching scores) and the real and estimated bids, yielding a spherical uncertainty set $\cal S$.

We can consider a spherical $(\delta, 0)$ uncertainty set to consist of a point estimate $S^0$ and a radius $\varepsilon$ limiting the $\cal L_2$ error from the point estimate $S^0$. Formally, we aim to solve the problem  $\argmax_{A \in \mathcal{A}}
    \inf_{S \in B_{\varepsilon}(S^0)}
    \USW(A, S)$,
%
where $B_{\varepsilon}(S^0) \doteq \{X \in \R^{n \times m} \ssep \norm{X - S^0}_F \leq \varepsilon\}$ denotes the $\varepsilon$ Frobenius-norm ball around $S^0$.
%
\begin{restatable}[
RAU
under Spherical Uncertainty]{theorem}{thmmaxuswrobust}\label{thm:maxuswrobust}
When $\cal S$ is a sphere,  
\begin{align*}
    \smash{\argmax_{A \in \mathcal{A}}}
    \inf_{S \in B_{\varepsilon}(S^0)}
    \USW(A, S) = \smash{\argmax_{A \in \mathcal{A}}} \, \smash{\USW(A, S^0)} \enspace,
\end{align*}
which can be computed in polynomial time. 
Furthermore, for all $A \in \cal A$,
\[
\abs{\USW(A, S^0) - \USW(A,S^\ast)} \leq \mathsmaller{\frac{\varepsilon \sqrt{K}}{n}} \enspace.
\]
\end{restatable}


One way of looking at \Cref{thm:maxuswrobust} is that $\AMM$ over spherical uncertainty sets provides no additional robustness guarantees over $\argmax_{A \in \cal A} \USW(A, S^0)$. 
Thus, it seems that in order to obtain meaningful robustness guarantees, we require both a limit to the total amount of variation in affinity scores (\Cref{thm:box}) as well as asymmetry between the noise on affinity scores (\Cref{thm:maxuswrobust}). We can explain these results intuitively; to decide how best to assign reviewers, we need to be able to make tradeoffs between assigning pairs with potentially high (but also potentially low) affinity or assigning pairs that have an average amount of affinity with higher certainty. Those tradeoffs are only meaningful if uncertainty varies across paper-reviewer pairs, and if there is a limited total amount of uncertainty.

With that intuition in mind, we generalize to the case of ellipsoidal  uncertainty sets. 
In a simple model, we might model affinity scores as multivariate Gaussians, 
which we explore in greater detail in the experiments of \Cref{sec:experiments}.
In this case, we obtain a mean vector\footnote{Technically, a sample from this multivariate Gaussian is a vector in $\R^{nm}$, and
must be reshaped into a matrix $S \in \R^{n \times m}$.
We will convert matrices into vectors in row-major order and ignore the distinction between matrices and vectors when convenient.} $\vec \mu \in [0,1]^{nm}$ and a positive semi-definite covariance matrix $\Sigma \in \mathbb{R}^{nm \times nm}$. 
Given a confidence level $1-\delta$, we create an uncertainty set 
\begin{align}
\label{eqn:normal_dist}
\cal S \doteq \left\{S \in \R^{nm} \ssep (S - \vec \mu)^T \Sigma^{-1} (S - \vec \mu) \leq \chi^2_{nm}(1-\delta) \right\} \enspace,
\end{align}
where $\chi^2_k$ is the inverse CDF of the $\chi^2$ distribution with $k$ degrees of freedom. If we assume no model error (which is not a safe assumption generally), then the true affinity scores are contained within $\cal S$ with probability at least $1-\delta$. 
We know also that $\chi^2_{nm}(1-\delta) \leq nm + 2\smash{\sqrt{nm \ln \mathsmaller{\frac{1}{\delta}}}} + 2\ln \frac{1}{\delta}$. 
We can see that the size of the uncertainty set depends only logarithmically on $\frac{1}{\delta}$, and thus we can trade off between $\delta$ and the $\cal L_1$ diameter of the uncertainty set.


While the Gaussian model
employed in \eqref{eqn:normal_dist}
makes a strong modeling assumption, we can use a validation set and a predictive model with provable tail bounds to obtain uncertainty sets that do not require any
such
assumptions. To accomplish this, we require a predictive model of affinity scores.  Venue organizers can then obtain $\cal L_2$ error bounds using validation data, and this will yield an ellipsoid due to sampling bias (only certain paper-reviewer pairs will be observed for any given venue). In this setting, the uncertainty set is not derived as a confidence interval of a probability distribution, but rather directly comes from a tail bound on total generalization error. Optimizing in this setting \emph{requires} using a robust approach like \problemname{}, and cannot be done with average case analysis.

\cyrus{Truncation a thorny issue for Gaussian model. Similar guarantees with sub-Gaussian or sub-gamma assumptions. Axis-aligned, a.k.a.\ nondiagonal $\Sigma$}

\newcommand{\Pap}{\bm{p}}
\newcommand{\Rev}{\bm{r}}
\newcommand{\PapUniv}{\cal P}
\newcommand{\RevUniv}{\cal R}

We now develop this predictive model in more technical detail. 
Suppose the true affinity score of some paper-reviewer pair is $f^{*}(p, r)$, and we have access to a predictive model $\hat{f}(p, r)$, perhaps learned on historical venues.
In practice, $\hat{f}$ predicts the affinity of reviewer $r$ for paper $p$ based on any information available prior to reviewer assignment, and 
the specific definition of affinity
is left to the
venue organizers.
For example, a venue may decide that affinity is best measured
via
reviewer bids, 
and they may
use historical data to train a
predictor $\hat{f}$ to predict missing bids from document-based similarity scores and keywords. 
Alternatively, venue organizers may decide that the ground truth affinity $f^\ast(p,r)$ should correspond to a meta-reviewer's judgment of review quality, and $\hat{f}$ can then be trained on historical data to predict these judgments.


We will take $\smash{\hat{S}_{p,r}} \doteq \smash{\hat{f}}(p, r)$ and $\smash{S^\ast_{p,r}} \doteq f^{*}(p, r)$. We assume that we can evaluate $\smash{\hat{f}}$ on all paper-reviewer pairs in the current venue, and potentially on pairs from historical venues as well. We may be able to sample $f^{*}(p, r)$ for some, but not all, pairs in the current venue and historical venues (the validation data).
We then probabilistically bound a weighted average of the square error between $\smash{\hat{f}}$ and $f^\ast$ in terms of an estimate of expected square error computed on the validation set. 
The details of the validation set vary by application, but the overall strategy will be to estimate the square error $\mathbb{E}[(S^\ast - \hat{S})^{2}]$, where $S^\ast$ and $\hat{S}$ are given by $f^{*}$ and $\smash{\hat{f}}$ on a random paper-reviewer pair. 
We show two such approaches, the first inductive, using historic or auxilliary data to form the validation set, and the second transductive, assuming a small random sample of true affinity scores (e.g., bids) can be queried within the current peer-review venue. We will need to define the notion of sampling a sequence of random variables \emph{conditionally independently without replacement}. 

\begin{definition}
A sequence of variables $x_1, x_2, \dots x_t$, where all $x_i \in \cal X$, is sampled conditionally independently without replacement from a distribution $\ProbDist$ with support $\cal X$ if the variables $x_i$ are sampled in order from $x_1$ to $x_t$, and for any $i \in \{1, \dots t\}$, $p_\ProbDist(x_i = x | x_1, \dots x_{i-1}) = \frac{p_\ProbDist(x_i = x)}{\int_{\cal X \setminus \{x_1, \dots x_{i-1}\}}p_\ProbDist(x_i = y)dy}$ for all $x \in \cal X \setminus \{x_1, \dots x_{i-1}\}$ and $p_\ProbDist(x_i = x | x_1, \dots x_{i-1}) = 0$ for $x \in \{x_1, \dots x_{i-1}\}$.
\end{definition}

\begin{restatable}[%
Ellipsoidal Uncertainty Sets from Inductive Predictors%
]{theorem}{thminductiveellipsoidaltailbounds}%
\label{thm:inductiveellipsoidal-tail-bounds}%
Let $\ProbDist'\!$ be a probability distribution over 
paper-reviewer pairs, 
and let $\ProbDist^{\PapUniv}\!$ and $\ProbDist^{\RevUniv}\!$ be distributions over papers and reviewers, respectively.  Assume 
that
$P$ and $R$
were drawn conditionally independently without replacement from $\ProbDist^{\PapUniv}\!$ and $\ProbDist^{\RevUniv}\!$, respectively.\footnote{Abusing notation slightly, we index the papers $p \in P$ and reviewers $r \in R$ so that $p$ can represent either a paper in $\PapUniv$ or an integer between $1$ and $n$.}
Suppose we sample $T$ paper-reviewer pairs $\{(p_i, r_i)\}_{i=1}^T$ conditionally independently without replacement from $\ProbDist'\!$, and these paper-reviewer pairs have true and estimated affinity scores $\{f^\ast(p_i, r_i)\}_{i=1}^T$ and $\{\hat{f}(p_i, r_i)\}_{i=1}^T$, respectively.

Let\[
\alpha(p, r) \doteq \frac{\Prob_{\Pap \distributed \ProbDist^{\PapUniv}} (\Pap = p) \Prob_{\Rev \distributed \ProbDist^{\RevUniv}}(\Rev = r)}{ \Prob_{(\Pap,\Rev) \distributed \ProbDist'\!}((\Pap,\Rev) = (p, r)) } 
\quad \& \quad \alpha_{\min} \doteq \inf_{p \in \PapUniv, r \in \RevUniv} \alpha(p, r)
\] denote (1) the \emph{probability 
ratio} of sampling $p$ from $\ProbDist^{\PapUniv}$ and $r$ from $\ProbDist^{\RevUniv}$ 
to 
sampling $(p, r)$ from $\ProbDist'\!$, and (2) the \emph{infimum probability ratio}, respectively.
Now construct the \emph{ellipsoid matrix} $\Sigma \in \R^{nm \times nm}$ as the diagonal matrix  such that 
$\Sigma_{pm+r,pm+r} = \alpha(p, r)$ for all $p \in P, r \in R$.

Then for any $\delta \in (0, 1)$, the ellipsoid
\[
\renewcommand{\!}{}
\mathcal{S} \! \doteq \! \left\{ \! S \! \in \! \R^{nm} \,\middle|\, \frac{1}{nm}(S \! - \! \hat{S})^\intercal \Sigma^{-1} (S \! - \! \hat{S}) \! \leq \! \smash{\underbrace{\frac{1}{T} \! {\sum_{i=1}^{T}} (f^\ast(p_i, r_i) \! - \! \hat{f}(p_i, r_i))^{2}}_{= \smash{\hat{\xi}}}} \! + \! \smash{\underbrace{\sqrt{\left( \frac{1}{T} \! + \! \frac{n\!+\!m}{nm\alpha_{\min}^{2}} \!\right) \!\frac{\ln \! \smash{\frac{1}{\delta}}}{2}}}_{= \eta}} \vphantom{\sqrt{\frac{1}{T}}} \, \right\} \vphantom{\underbrace{\frac{1}{T}}_{\xi}} 
\]
is a $(\delta, 0)$ uncertainty set,
where $\smash{\hat{\xi}} 
$ denotes the \emph{empirical square error} of our estimated scores, and $\eta 
$ denotes the \emph{excess error bound} due to sampling.
\end{restatable}


 
Departing from the standard reviewer assignment setup,
\Cref{thm:inductiveellipsoidal-tail-bounds} assumes that both the historic data and the current venue are random.
In particular, historic paper-reviewer pairs are sampled from $\ProbDist'$ (modeling the historic data generation process), and papers and reviewers for the current venue are sampled from $\ProbDist^\PapUniv$ and $\ProbDist^\RevUniv$ (modeling the processes by which 
papers are submitted and reviewers volunteer).
We then construct
$\alpha(p, r)$
to reweight square error on the current venue to match expected square error on $\ProbDist'$ (i.e., we use importance sampling to calibrate expectations over $\ProbDist'$ versus those over $\ProbDist^\PapUniv$ and $\ProbDist^\RevUniv$). 
For example, $\ProbDist'$ reflects all elements of historic data generation, most importantly the availability of historic data from multiple venues with different focuses. We might then use the (relatively stable) topic areas of papers and reviewers to model 
$\ProbDist^\PapUniv$ and $\ProbDist^\RevUniv$\!, and thus $\alpha(p, r)$ 
reflects the ratio of the popularity of $p$ and $r$'s topic areas in the current venue to historic venues.
%

We show a similar result in the transductive setting. Instead of constructing a predictive function from historical data, we generalize a small set of known affinities for the current venue to the unknown affinities for the same venue. Note that $\ProbDist'$ now reflects the process by which we obtain samples for $(p_i, r_i)$ pairs from the current venue, rather than from historical venues.

\begin{restatable}[
Ellipsoidal Uncertainty Sets from Transductive Predictors%
]{theorem}{thmtransductiveellipsoidaltailbounds}%
\label{thm:transductiveellipsoidal-tail-bounds}%
Suppose we sample $T$ paper-reviewer pairs $\{(p_i, r_i)\}_{i=1}^T$ conditionally independently without replacement from $\ProbDist'\!$, and these paper-reviewer pairs have true and estimated affinity scores $\{f^\ast(p_i, r_i)\}_{i=1}^T$ and $\{\hat{f}(p_i, r_i)\}_{i=1}^T$, respectively.
Let 
\[
\alpha(p, r) \doteq \vphantom{\frac{o}{2}}\smash{\frac{\smash{(nm)^{-1}}}{ \Prob_{(\Pap,\Rev) \distributed \ProbDist'}((\Pap,\Rev) = (p, r)) }} 
\] 
denote the \emph{probability 
ratio} between
sampling 
$(p, r)$ uniformly at random and
sampling
$(p, r)$ from $\ProbDist'\!$,
and construct the \emph{ellipsoid matrix} $\Sigma \in \R^{nm \times nm}$ as the diagonal matrix  such that $\Sigma_{pm+r,pm+r} = \alpha(p, r)$ for all $p \in P, r \in R$.

Then for any $\delta \in (0, 1)$, the ellipsoid
\[
\renewcommand{\!}{}
\mathcal{S} \doteq \left\{ S \in \R^{nm} \,\middle|\, \frac{1}{nm}(S - \hat{S})^\intercal \Sigma^{-1} (S - \hat{S}) \leq \smash{\underbrace{\frac{1}{T} \vphantom{\sum_{\cdot}}\smash{\sum_{i=1}^{T}} (f^\ast(p_i, r_i) \! - \! \hat{f}(p_i, r_i))^{2}}_{= \hat{\xi}}} + \smash{\underbrace{\sqrt{\frac{\ln \smash{\frac{1}{\delta}}}{2T}}}_{= \eta}} \vphantom{\sqrt{\frac{1}{T}}} , \right\} \vphantom{\underbrace{\frac{1}{T}}_{\eta}}
\]
is a $(\delta, 0)$ uncertainty set, where $\smash{\hat{\xi}}$ denotes the \emph{empirical square error} of our estimated scores, and $\eta$ denotes the \emph{excess error bound} due to sampling.

\end{restatable}

The transductive result can be straightforwardly applied to many different contexts in which venue organizers can
solicit samples of $f^\ast$ 
on $(p_i, r_i)$ pairs from the current venue, rather than historical data. 
This information must be obtained prior to assigning the majority of reviewers. For example, 
organizers could define $f^\ast$ as 
a reviewer's hypothetical bid
and \Cref{thm:transductiveellipsoidal-tail-bounds} then requires soliciting a small number of bids 
to estimate the error of $\hat{f}$. 
Similarly, $f^\ast$ 
could correspond to meta-reviewer judgments of review quality, accomplished by opting for a two-stage reviewing process, in which the reviews and feedback generated in the first stage are used to estimate the error of $\hat{f}$. 
These definitions of $f^\ast$ are costly to sample, but organizers can still efficiently target sophisticated affinity models by solving \problemname{} over the uncertainty sets of \Cref{thm:transductiveellipsoidal-tail-bounds}.

We naturally ask the question, ``How many samples are sufficient to obtain a sharp confidence bound?''
Observe that, by \cref{prop:sandwich}, the gap between adversarial and true welfare is $\frac{L}{n}$, where $L$ denotes the $\cal L_1$ diameter of $\cal S$.
For the ellipsoidal uncertainty set of \cref{thm:transductiveellipsoidal-tail-bounds}, $\frac{L}{n} \leq 2m\sqrt{\alpha_\text{max} (\hat{\xi} + \eta )}$, where $\alpha_{\max} \doteq \sup_{p \in \PapUniv, r \in \RevUniv} \alpha(p, r)$.
Furthermore, the empirical square error $\hat{\xi}$ converges to some $\xi$ as $T$ increases, thus we need only select $T \in \LandauOmega \bigl(
\smash{\frac{\log \frac{1}{\delta}}{\xi^{2}}} \bigr)$ samples to ensure that the uncertainty set is constant-factor optimal, at which point the welfare gap is $\LandauO( m\sqrt{\alpha_\text{max} \xi })$, which is also optimal to within constant factors. 
Notably,
the sufficient sample size $T$ is \emph{independent of the venue size} (i.e., $n$ and $m$), thus the added burden of soliciting these extra bids is negligible.
We also see that the fundamental limitation of this method is the 
average square error 
$\xi$, which depends on the predictor, the venue, and the sampling distribution $\ProbDist'$.
It is thus paramount to use predictors for which this quantity will be small.
Fortunately, this is often the case, as many predictive models are explicitly trained to minimize $\mathcal{L}_{2}$ error on some task, which motivates the choice of our ellipsoidal uncertainty sets.
Note that while this argument pertains to \cref{thm:transductiveellipsoidal-tail-bounds}, one can argue similarly for the necessary size of the validation set to ensure $\eta = \LandauO(\xi)$ 
in \cref{thm:inductiveellipsoidal-tail-bounds}.

Finally, we note that it is possible to extend either result under less favorable (more realistic) assumptions about the sampling process using the $\gamma$ parameter ($\mathcal{L}_{1}$ error) of our uncertainty set construction. In particular, either result produces a ($\delta, T\gamma$) uncertainty set if $\hat{s}$, $s^\ast$, and the associated ($p, r$) pairs are subject to \emph{adversarial corruption} of $T\gamma$ 
of the validation set samples drawn from $\ProbDist'$, which has immediate applications in privacy, adversarial robustness, and various notions of strategy-proofness.
Furthermore, to model more complicated and potentially not fully understood distribution shift, we 
obtain via Bennett's inequality \cite{bennett1962probability} a $\Bigl( \delta + \delta', T\gamma + \mathsmaller{\frac{1}{3}}\ln \frac{1}{\delta'} + \sqrt{2T\gamma(1-\gamma)\ln \frac{1}{\delta'}}\, \Bigr)$ uncertainty set 
if the validation set is
instead drawn from some $\ProbDist''$ such that $\TVD(\ProbDist', \ProbDist'') \leq \gamma$. \cyrus{The key idea behind this observation is that you change samples according to a $\gamma$ Bernoulli, so you can then use Bernoulli tail bounds.}
\cyrus{TVD works, anything for earth-movers?}
\cyrus{Is there anything we could cite for “privacy, adversarial robustness, and various notions of strategy-proofness?”}

\cyrus{Example: reweighting reviewers AND papers by sector: independently and dependently (i.e., bias towards within-sector samples).} \justin{You bias sampling distribution towards things in own sector, or things with a close similarity score, or could be weighted by the inverse log rank}

\cyrus{This result should be compared to the Gaussian-based ellipsoid, ...}

\justin{Make a comment about how probabilistic models (like Gaussian above) can make wrong modeling assumptions, but this model cannot go wrong.}

\subsection{Compositional Rules}
\label{subsec:compositional}
\justin{We can start here by saying that we can take intersections of uncertainty sets, so we don't have to overly rely on any 1 uncertainty set. And let's probably say that if you've got $\gamma, \delta$ uncertainty sets you can take the intersection of them all.}

We may often have more complicated uncertainty sets than the simple geometries described in the previous sections. For example, we can intersect the constraints of the unit hypercube with an ellipsoidal uncertainty set as described in \Cref{subsec:ellipsoid}. This produces a \emph{truncated ellipsoid}, a common construction that we will see again in 
\Cref{sec:experiments}.

%
\begin{restatable}[Uncertainty Set Intersection]{lemma}{thmuncsetint}
\label{lemma:uncsetint}
Suppose each $\mathcal{S}_{i}$ for $i \in [k]$ is a $(\smash{\vec \delta_{i}}, 0)$ uncertainty set. 
Then $\mathcal{S}_{\cap} \doteq \bigcap_{i=1}^{k} \mathcal{S}_{i}$ is a
$(\norm{\smash{\vec \delta}}_{1},0)$
uncertainty set.
\cyrus{Future work: $\bm{\gamma}_{\max}$ for axis-aligned rectangular uncertainty sets, and in general uncertainty sets that always intersect at obtuse angles (e.g., rectangular + contained ellipsoid).}
\end{restatable}
%
We can also use \Cref{lemma:l1-error-contract} to convert $(\delta, \gamma)$ uncertainty sets to larger $(\delta, 0)$ uncertainty sets. 
%
\begin{restatable}[$\mathcal{L}_{1}$ Error Terms]{lemma}{lemerrorcontract}
\label{lemma:l1-error-contract}
If $\mathcal{S}$ is a $(\delta, \gamma)$ uncertainty set, then for any $\eta \in [0, \gamma]$, it holds that the \emph{Minkowski sum}
\[
\cal S' \doteq \mathcal{S} + \{ \vec{s} \in \R^{nm} \ssep \norm{\vec{s}}_{1} \leq \eta \}  = \{\vec x + \vec s \ssep \vec x \in \cal S, \norm{\vec s}_1 \leq \eta \}
\]
is a $(\delta, \gamma - \eta)$ uncertainty set.
\end{restatable}


We can apply 
\cref{lemma:uncsetint,lemma:l1-error-contract}
sequentially
to intersect arbitrary $(\delta, \gamma)$ uncertainty sets. We first expand 
them
via \cref{lemma:l1-error-contract} to obtain larger 
$(\vec \delta_i, 0)$ uncertainty sets, 
and we then apply \Cref{lemma:uncsetint} to obtain their intersection.
%
%
%
%
%

We can also apply these results to the uncertainty sets previously described in \Cref{sec:uncertainty_set_types}.
For example, the intersection of multiple axis-aligned, hyperrectangular constraints produces an axis-aligned hyperrectangle. This may occur when structural constraints defined by the venue (e.g., hard upper and lower bounds defined based on topic overlap) intersect with per-pair error bounds.
Similarly, we might consider cases with two intersecting ellipsoidal error bounds derived from two different estimators using \cref{thm:inductiveellipsoidal-tail-bounds,thm:transductiveellipsoidal-tail-bounds}. 
This intersection is
uninteresting if the two ellipsoids have the same centroid and one is strictly smaller than the other,
but if these ellipsoids have different centroids (as when the estimators have different biases) their intersection can be quite beneficial.

\cyrus{Note also that while the Minkowski sum and difference operations on $\mathcal{L}_{1}$ balls may or may not have efficient direct representations (exponentially many additional linear or quadratic constraints, differences with hyperplane constraints Translate to hyperplane constraints, and sums of convex constraints produce convex constraints, fit into disciplined CP.
NB: $A + \mathcal{B}_{1,r} \leq A + \mathcal{B}_{\infty,r}$, useful for hyperrectangular constraints $A$. I think $A \subseteq A + \mathcal{B}_{\infty,r} \subseteq B$ for ellipsoids $A$, $B$ where $B$ expands each radius by $r$? 
}

\section{Robust Reviewer Assignment (\algoname{})}
\label{sec:approx_algo}

We now present a general purpose algorithm for approximately solving the \problemname{} problem over convex uncertainty sets, as long as the adversarial (worst-case) welfare can be computed in polynomial time.
We first show in \Cref{thm:hardness} that \problemname{} is NP-hard in general for convex uncertainty regions of this type. 



\begin{restatable}[Hardness of RAU]{theorem}{thmhardness}\label{thm:hardness}
    \problemname{} is NP-hard over a convex uncertainty set $\mathcal{S}$, even for $\mathcal{S}$ with a polynomial-time adversary. In particular, \problemname{} remains NP-hard even when $\mathcal{S}$ is restricted to bounded polytopes formed by intersections of polynomially many halfspaces.
\end{restatable}
 

\subsection{Robust Reviewer Assignment}

Due to this hardness result, we outline an approach
to approximately solve \problemname{}
efficiently for convex uncertainty sets with polynomial-time adversaries. We start by allowing \emph{fractional} (rather than binary) assignments. 
We then apply \emph{supergradient ascent} (analogous to subgradient descent) to approximate $
    \argmax_{A \in \smash{\tilde{\mathcal{A}}}}
    \inf_{S \in \mathcal{S}}
    \USW(A, S)
$,
where $\tilde{\mathcal{A}}$ 
is the convex closure of the feasible set of discrete allocations $\cal A$ from \Cref{sec:problem}. 
When the supergradient ascent algorithm terminates, we randomly round the assignment to a binary assignment.

%

\begin{algorithm}
    \scalebox{0.82}{\begin{minipage}{1.219512195\textwidth}
    \begin{algorithmic}[1]
        \REQUIRE {Error tolerance $\varepsilon$, 
        supergradient norm bound $\lambda$, uncertainty set $\mathcal{S}$, constrained allocation space $\tilde{\mathcal{A}}$, total review load $K = \sum_{p \in P}k_p$ (i.e., the total number of assignments required)}
        \STATE Initialize $S^{(0)} \in \mathcal{S}$ arbitrarily
        \STATE $\displaystyle A^{(0)} \gets \argmax_{A \in \mathcal{A}} \USW(A, S^{(0)})$ \COMMENT{Initialize $A^{(0)}$ to optimize $S^{(0)}$ (via LP reduction)}%
        \STATE $\hat{A} \gets  A^{(0)}; \hat{w} \gets -\infty$ \COMMENT{Maintain best allocation $\hat{A}$ and adversarial welfare $\hat{w}$}%
        \STATE $T \gets \big\lceil 2K( \frac{\lambda}{\varepsilon})^2 \big\rceil ; \alpha \gets \frac{\varepsilon}{\lambda^2}$ \COMMENT{Compute sufficient \emph{step count} $T$ and \emph{step size} $\alpha$}%
        \FOR {$t \in \{1, 2, \dots T\}$}
            \STATE $\displaystyle S^{(t)} \gets \arginf_{S \in \mathcal{S}} \USW(A^{(t-1)}, S)$
            \COMMENT{Adversary selects $S^{(t)}$ from $\cal S$
            }%
            
            \IF{$\USW(A^{(t-1)}, S^{(t)}) > \hat{w}$\COMMENT{Update $\hat{A}$ if adversarial welfare beats previous best\!}}%
                \STATE $\hat{A} \gets A^{(t-1)}; \hat{w} \gets \USW(A^{(t-1)}, S^{(t)})$ 
            \ENDIF
            
            \STATE $\displaystyle A^{(t)} \gets A^{(t-1)} + \alpha \nabla_{A^{(t-1)}} \USW(A^{(t-1)}, S^{(t)})$ 
            \COMMENT{Update allocation with a supergradient step}%
            \STATE $\displaystyle A^{(t)} \gets \argmin_{A \in \smash{\tilde{\mathcal{A}}}} \lVert A - A^{(t)} \rVert_2$ \COMMENT{$\cal L_2$ project onto feasible allocation set $\tilde{\mathcal{A}}$}
        \ENDFOR
        \RETURN $\ROUND(\hat{A})$ \COMMENT{Sample integral assignment}
    \end{algorithmic}
    \end{minipage}
    }
   \caption{Robust Reviewer Assignment (\algoname{})}

    \label{alg:ouralg}
\end{algorithm}

In particular, we present \cref{alg:ouralg}, termed \algoname{}. 
\algoname{} applies an iterative adversarial optimization strategy to the objective.
In each iteration $t$, we take an \emph{adversary step}, which identifies the pessimal $S^{(t)}$ given assignment $A^{(t-1)}$. We then take a \emph{gradient ascent step} from $A^{(t-1)}$ to $A^{(t)}$ assuming the score matrix remains fixed at $S^{(t)}$, followed by a \emph{projection step}, which ensures $A^{(t)}$ remains feasible (i.e., does not violate any constraints on assignments). The gradient ascent step is valid since the gradient $\nabla_A \USW(A, \arginf_{S \in \cal S} \USW(A, S))$ is an element of the supergradient $\nabla_A \inf_{S \in \cal S} \USW(A, S)$. 

The number of iterations required to prove convergence depends on the \emph{gradient norm bound} $\lambda$, which is the smallest term such that $\norm{\nabla_A \USW(A, S)}_2 \leq \lambda$ for all $A$ and $S$. We approximate the maximin optimal continuous matrix $\AMMC $  within an error of $\varepsilon$ in number of iterations polynomial in $\lambda$, $\frac{1}{\varepsilon}$, and the total review load $K$. The time complexity of \algoname{} also depends on the \emph{adversarial minimization} and \emph{projection} steps, but so long as these take polynomial time, then so too does \algoname{}. We state the convergence results in \Cref{prop:subgradconv}. The proof applies standard convergence results for subgradient descent \cite{shor1985minimization}.
%
%
%
\begin{restatable}[Supergradient Ascent Efficiency]{proposition}{propsubgradconv}\label{prop:subgradconv}
Let $\lambda$ denote an upper bound on the $\mathcal{L}_2$ norm of the supergradient elements $\nabla_{A} \USW(A, S)$ used in \algoname{}. 
The supergradient ascent component of \algoname{} converges to within $\varepsilon$ of the maximin optimal continuous assignment $\AMMC$ in $\lceil 2K ( \frac{\lambda}{\varepsilon} )^{2} \rceil$ iterations. \algoname{} runs in time $\LandauO\big( 2KC\smash{( \frac{\lambda}{\varepsilon} )^{2}} \big)$, where $C$ is the
time cost
 of one 
 adversary and projection step.
\end{restatable}

Although the bound on the number of iterations can be quite large, it proves the convex relaxation of \problemname{} is solvable in polynomial time, as long as the adversary and projection steps can be solved in polynomial time and $\lambda$ is bounded. In addition, the required number of iterations until convergence will typically be much smaller in practice. 

The complexity result improves in the case of (truncated) ellipsoidal uncertainty sets.
The \emph{adversarial minimization} step requires polynomial time under truncated ellipsoidal uncertainty sets, as it is a \emph{linear objective} under \emph{convex quadratic constraints} (and box constraints), which is a \emph{second-order conic program}, and the \emph{projection step} always requires polynomial time, 
as it is a \emph{convex quadratic} objective under \emph{linear} constraints (i.e., the assignment constraints $\tilde{\mathcal{A}}$).
The bound $\lambda$ can be difficult to compute in the general case, but we show $\lambda$ is typically well-bounded in the case of truncated ellipsoidal uncertainty sets. 

\begin{restatable}[Supergradient Ascent Efficiency under Ellipsoidal Uncertainty]{corollary}{corellipsoidalconvergence}\label{cor:ellipsoidalconvergence}
For a truncated ellipsoidal uncertainty set, the supergradient ascent component of \algoname{} converges to within $\varepsilon$ of the maximin optimal continuous assignment $\AMMC$ in $\LandauO\left(\frac{2Km}{n\varepsilon^2}\right)$ iterations. 
\end{restatable}




Finally, we can round using the extended Birkhoff von Neumann decomposition sampling algorithm \cite{gandhi2006dependent,budish2009implementing,jecmen2020random}. This sampling algorithm generates an integral sample $\smash{\hat{A}'}$ from the distribution defined by the continuous assignment matrix $\smash{\hat{A}}$. The sample $\smash{\hat{A}'}$ still satisfies the constraints of $\tilde{\mathcal{A}}$, and $\Expect_{A'}[\smash{\hat{A}}] = \hat{A}$. 
The time complexity of this sampling algorithm is $\LandauO(mn(m+n))$, which is typically negligible compared to the complexity of supergradient ascent.

\subsection{Maximin and Integrality Gaps}
\label{subsec:gaps}

\cyrus{Most of this section is implicitly w.h.p., subject to some tail bounds. Clarify this, then proceed as though bounds hold.}

\cyrus{Todo: consider case where $S \in \mathcal{S}$, and also $\norm{\cdot}_{\infty}$-Hausdorff distance, i.e., $\Delta_{\norm{\cdot}_{\infty}}(S, \mathcal{S}) \min \norm{\cdot}_{\infty} \in \mathcal{S} \leq \gamma$}

Aside from 
the optimization error of \algoname{} (\Cref{alg:ouralg}), 
there are two more error sources: \emph{maximin error} for working under uncertainty, and \emph{rounding error}.
The integrality gap of \problemname{} can be quite large; similarly, the $\cal L_1$ difference between a rounded assignment $A'$ and a continuous assignment $\tilde{A}$ may be quite large as well. Surprisingly, we show in this section that although the integrality gap  is large, with high probability this does not translate to a large amount of suboptimality in the \emph{true} $\USW$ of the rounded solution $A'$. Intuitively, whenever the maximin optimal continuous solution $\AMMC$ has a high $\cal L_1$ distance from any valid binary integer assignment, 
the decisions made during rounding cancel out on average, and have relatively little impact on the true welfare of the assignment. 

\begin{restatable}[$\cal L_1$ Distance to Integral Solution]{proposition}{integ_gap_lb}
\label{prop:integ_gap_lb}
Suppose
an unrounded
assignment $\tilde{A}$ and a randomized rounding $A'$ of $\tilde{A}$
such that $\mathbb{E}[A'] = A$.
Then the expected $\cal L_1$ deviation of the assignment due to rounding obeys
\[
\mathbb{E}_{A'}\left[\norm{A' - \smash{\tilde{A}}}_{1} \right] =2\left(\norm{\smash{\tilde{A}}}_1 - \norm{\smash{\tilde{A}}}^2_2\right)
\leq nm - 2\norm*{\mathsmaller{\frac{1}{2}} - \smash{\smash{\tilde{A}}}}_{1}
\enspace.
\]
\end{restatable}

Although the assignments may need to be rounded quite significantly, \Cref{thm:gaps} shows that the rounded assignment produced by \algoname{} has near-optimal true welfare (in expectation). \oldCref{thm:gaps} also allows for $\varepsilon$ error in the discrete/continuous maximin assignments. For this, we define $A^\varepsilon \in \cal A$, representing any $\varepsilon$-optimal discrete solution to $\problemname{}$ and $\smash{\tilde{A}^\varepsilon} \in \smash{\tilde{\cal A}}$, any $\varepsilon$-optimal continuous solution to $\problemname{}$. $A^\varepsilon$ and $\smash{\tilde{A}^\varepsilon}$ are formally defined by the properties
\begin{align*}
\renewcommand{\!}{\hspace{-0.1ex}}
\scalebox{0.993}{\ensuremath{\displaystyle
\max_{A \in \cal A}\!\inf_{S \in \cal S}\!\USW(A, S) \! - \! \inf_{S \in \cal S}\!\USW(A^\varepsilon\!, S) \! \leq \! \varepsilon 
\,\ \& \ 
\max_{A \in \tilde{\cal A}}\!\inf_{S \in \cal S}\!\USW(A, S) \! - \! \inf_{S \in \cal S}\!\USW(\tilde{A}^\varepsilon\!, S) \! \leq \! \varepsilon \!\enspace.
}}
\end{align*}%
\begin{restatable}[Maximin and Integrality Gaps in Welfare]{theorem}{thmgaps}
\label{thm:gaps}
Suppose $\mathcal{S}$ is a $(\delta,\gamma)$ uncertainty set with $\cal L_1$ diameter $L$. Let $A^\varepsilon$ denote an $\varepsilon$-optimal discrete \problemname{} solution, and $\tilde{A}^\varepsilon$ denote an $\varepsilon$-optimal continuous \problemname{} solution.
Let $A'$ denote the random variable that
arises from
applying the randomized rounding procedure $\ROUND$ to $\tilde{A}^\varepsilon$, and assume that $\ROUND$
preserves expectation, i.e., $\Expect_{A'}[A'] = \tilde{A}^\varepsilon$. 
Suppose also that the true affinity scores are $S^\ast$, and denote the optimal solution $\begin{aligned}A^\ast \doteq \smash{\argmax_{A \in \cal A}} \, \USW(A, S^\ast)\end{aligned}$. The following then hold. 
\vspace{.5em}
%
%
\cyrus{For a sharper result, define $L \doteq \sup_{A} \sup_{S, S' \in \mathcal{S}} A \cdot (S - S')$.  By H\"older's inequality, $L \leq [l1 diameter]$, and result still holds.
Sharper still is more of a radius, $L \doteq \sup_{A} \sup_{S \in \mathcal{S}} \abs{A \cdot (S - S^{\ast})}$.
But with MM methods, can improve to $L \doteq \sup_{A} \inf_{S \in \mathcal{S}} \sup_{S' \in \mathcal{S}} A \cdot (S - S')$.
Using a fixed affinity score matrix $S$, we would instead have guarantee with $L \doteq \sup_{A} \sup_{S' \in \mathcal{S}} A \cdot (S - S')$.
So maximin is at least as good at being robust within $\mathcal{S}$ as the ``ideal'' fixed scores, without having to compute said ideal?
}
\begin{enumerate}
\item Maximin Gap: $\Prob \bigl( \USW(A^\ast, S^\ast) - \USW(A^\varepsilon, S^\ast) > \varepsilon + \frac{2\gamma + L}{n} \bigr) < \delta$.

\item Expected Regret of \algoname{}: $ \Prob \bigl( \USW(A^\ast, S^\ast) - \Expect_{A'}[\USW(A',S^\ast)] > \varepsilon +\frac{2\gamma + L}{n} \bigr) < \delta$.

\item 
Probabilistic Regret of \algoname{}: 
$\Prob \bigl( \USW(A^\ast, S^\ast) - \USW(A',S^\ast) > \frac{\varepsilon + (2\gamma + L)/n}{\delta'} \bigr) < \delta' + \delta $.

\end{enumerate}

\end{restatable}

Since $A'$ is the result of \algoname{}, \Cref{thm:gaps} directly bounds the expected and probabilistic regret of \algoname{}. Note that the distribution for the probabilistic regret is over the randomness of the rounding procedure, while all bounds in \Cref{thm:gaps} are probabilistic with respect to
$\delta$ of the uncertainty set $\cal S$.


\cyrus{More discussion:
MM gap fundamentally matters, integrality gap only matters for deterministic guarantees?}

\section{Experiments}
\label{sec:experiments}

We create uncertainty sets using five years of ICLR data, by constructing an asymmetric multivariate Gaussian and taking a confidence interval (as outlined by the example in \eqref{eqn:normal_dist}). We compare the adversarial and average-case performance of our approach to the adversarial and average-case performance of four commonly used baselines. We then demonstrate the importance of optimizing for the adversarial case. 
Code and data are publicly available on Github.\footnote{\url{https://github.com/justinpayan/RAU}}

\subsection{Baseline Comparison}
\label{subsec:comparisontobaseline}

\begin{table*}
      \centering
        \caption{Adversarial and average welfare (mean $\pm$ standard deviation) for  the na{\"i}ve \LP{}, FairFlow, PeerReview4All, FairSequence, and \algoname{} methods
        on 
        five
        ICLR conferences.
        Welfare is scaled by $100$ for ease of comparison. 
        Adversarial welfare is consistently highest (\textbf{bold}) using \algoname{}, except for 2018, which is within one standard deviation. 
        }
        \renewcommand{\toprule}{\hline}
        \renewcommand{\midrule}{\hline}
        \renewcommand{\bottomrule}{\hline}
        \setlength{\tabcolsep}{0.4em} 
        \renewcommand{\arraystretch}{0.93}
        \rowcolors{5}{white!82!gray}{white}
        \small

    \small
    \setlength{\tabcolsep}{0.2em}
    \renewcommand{\,}{\hspace{0.333pt}}
    \newcommand{\tns}{\hspace{-0.2pt}}
    \newcommand{\pmsmall}{\raisebox{0.8pt}{\ensuremath{\mathsmaller{\,\pm\,}}}}
    \begin{tabular}{ccc @{\hspace{3.75pt}}|@{\hspace{3.75pt}} ccccc @{\hspace{3.75pt}}|@{\hspace{3.75pt}} ccccc}
        \toprule
            & & & & & & & & & & & & \\[-2.0ex] 
            \multirow{2}{*}{\!\!\scalebox{0.95}[0.99]{\textbf{Year}}\!\!\tns\tns} & \multirow{2}{*}{\!\textbf{m}\!} & \multirow{2}{*}{\!\textbf{n}\!} & & \multicolumn{3}{c}{\hspace{-1cm}\textbf{Adversarial USW $\bm{\cdot}$ $\bm{100}$}\hspace{-1cm}\null} & & \multicolumn{5}{c}{\textbf{Average USW $\bm{\cdot}$ $\bm{100}$}} \\[0.75pt]
             & & & \bf LP & \bf FF & \!\scalebox{0.92}[0.98]{\bf PR4A}\! & \bf FS & \!\scalebox{0.92}[0.98]{\bf \algoname{}}\! & \bf LP & \bf FF & \!\scalebox{0.92}[0.98]{\bf PR4A}\! & \bf FS & \bf \!\!\scalebox{0.92}[0.98]{\algoname{}}\!\! \\
            \toprule
            & & & & & & & & & & & & \\[-2.3ex] 
            \scalebox{0.94}[0.98]{$2018$} & $1657$ & $\hphantom{1}546$ & $\bm{17} \pmsmall 3$ & $7 \pmsmall 3$ & $\bm{17} \pmsmall 3$ & $16 \pmsmall 3$ & $16 \pmsmall 3$ & $\bm{179} \pmsmall 2$ & \!\scalebox{0.915}[0.96]{$134 \tns \pmsmall \tns 12$}\! & $177 \pmsmall 2$ & $177 \pmsmall 2$ & $160 \pmsmall 4$ \\  
            \scalebox{0.9}[0.98]{$2019$} & $2620$ & $\hphantom{1}851$ & $22 \pmsmall 2$ & $12 \pmsmall 2$ & $22 \pmsmall 2$ & $22 \pmsmall 2$ & $\bm{27} \pmsmall 3$ & $\bm{184} \pmsmall 1$ &  $139 \pmsmall 9$ & $\bm{184} \pmsmall 1$ & $183 \pmsmall 1$ & $161 \pmsmall 3$ \\
            \scalebox{0.94}[0.98]{$2020$} & $4123$ & $1327$ & $17 \pmsmall 2$ & $11 \pmsmall 2$ & $18 \pmsmall 2$ & $17 \pmsmall 2$ &  $\bm{23} \pmsmall 2$ & $\bm{187} \pmsmall 1$ & $158 \pmsmall 8$ & $\bm{187} \pmsmall 1$ & $186 \pmsmall 1$ & $166 \pmsmall 5$ \\

            \scalebox{0.94}[0.98]{$2021$} & $4662$ & $1557$ & $23 \pmsmall 2$ & $18 \pmsmall 2$ & $23 \pmsmall 2$ & $23 \pmsmall 2$ & $\bm{33} \pmsmall 3$ & $\bm{192} \pmsmall 1$ & $177 \pmsmall 2$ & $\bm{192} \pmsmall 1$ & $191 \pmsmall 1$ & $174 \pmsmall 6$ \\
            \scalebox{0.94}[0.98]{$2022$} & $5023$ & $1576$ & $28 \pmsmall 2$ & $23 \pmsmall 2$ & $28 \pmsmall 2$ & $27 \pmsmall 2$ & $\bm{38} \pmsmall 2$ & $\bm{191} \pmsmall 1$ & $177 \pmsmall 1$ & $190 \pmsmall 1$ & $190 \pmsmall 1$ & $172 \pmsmall 3$ \\
            \bottomrule
            \end{tabular} \\
\label{tab:iclr}
\end{table*}

We examine the ability of our approach to robustly maximize utilitarian welfare in the case of \emph{truncated} ellipsoidal uncertainty sets $\mathcal{S}$. This experiment is meant to mimic real-world conference scenarios, and thus we do not assume access to a ground truth affinity matrix. Consequently, we can only compare adversarial and average-case welfare for \algoname{} vs. our baselines, but not true welfare.

We use the OpenReview API
to collect all papers submitted (both accepted and rejected) to  five recent iterations of ICLR (2018--2022).
Following recent work, 
we use the pool of authors for each year as the reviewer pool, since we do not have access to the true reviewer identities for these conferences. The number of reviewers and papers for each
conference year
is shown in \Cref{tab:iclr}.

For each author in each year, we collect the multiset of keywords from papers the author submitted to ICLR in the current or previous years. 
We then follow a procedure similar to that of AAAI 2021 \cite{leyton2022matching} to convert keywords into a mean  vector $\vec \mu \in \R^{nm}$, and we also construct a covariance matrix $\Sigma \in \R_{\geq 0}^{nm \times nm}$ for paper-reviewer affinity scores. 
Vector $\vec p$ is set so  $\vec p_i$ is an indicator for keyword $i$ on paper $p$. Vector $\vec r$ is initially set so  $\vec r_i$ is the number of times the keyword $i$ appears on a paper written by that reviewer in this or previous years' conferences. We then modify the values (but not the ordering) of $\vec r$ such that the minimum non-zero value is $0.2$, the maximum value is $1$, and the remaining non-zero values are evenly-spaced between $0.2$ and $1$. 
Let $\vec \lambda \in \R^V$ be such that $\vec \lambda_i = \left(\frac12\right)^{i-1}$. $\Sorted$ represents the function that sorts values of a vector in decreasing order, $M_{\vec p}$ and $M_{\vec r}$ denote the number of non-zero entries in $\vec p$ and $\vec r$ respectively, and $Z = \sum_{i=1}^{M_{\vec p}} \left(\frac12\right)^{i-1}$. 
We set $\vec \mu_{pr} = \frac{\vec \lambda \cdot \Sorted(\vec p \circ \vec r)}{Z}$, $\Sigma_{pm+r,pm+r} = (M_{\vec p}M_{\vec r})^{-2}$, and all off-diagonal entries of $\Sigma$ are $0$. 
This procedure  was chosen to roughly mirror the procedure used by AAAI 2021.
For each year of ICLR, we set the uncertainty set $\mathcal{S}$ for robust optimization to be the $95\%$ confidence interval for the distribution $\mathcal{N}(\vec \mu, \Sigma)$ (as in \eqref{eqn:normal_dist}), intersected with the unit hypercube (\Cref{lemma:uncsetint}).

For each year of ICLR, we sample without replacement $60\%$ of the reviewers and $60\%$ of the papers $100$ times, to produce more data for statistical robustness of our experiments. 
We assume that all papers require $3$ reviews and all reviewers can review up to $6$ papers. 
There are no conflicts of interest. 
We then run \algoname{} using our calculated confidence region $\mathcal{S}$, and compare against the assignment given by $\argmax_{A \in \mathcal{A}} \USW(A, \vec \mu)$ (the ``\LP{}'' solution); that is, the n{a\"i}ve solution that optimizes for the mean score. We also compare against three  baselines commonly used for real conferences, FairFlow \cite{kobren2019paper}, PeerReview4All \cite{stelmakh2019peerreview4all}, and FairSequence \cite{payan2022order}.
The results of these experiments are shown in \Cref{tab:iclr}.  All baselines lag far behind \algoname{} in adversarial performance.

\subsection{The Importance of Adversarial Analysis}
\label{subsec:worstcase}

\begin{figure}[t]
    \centering
    \!\includegraphics[width=.495\textwidth]{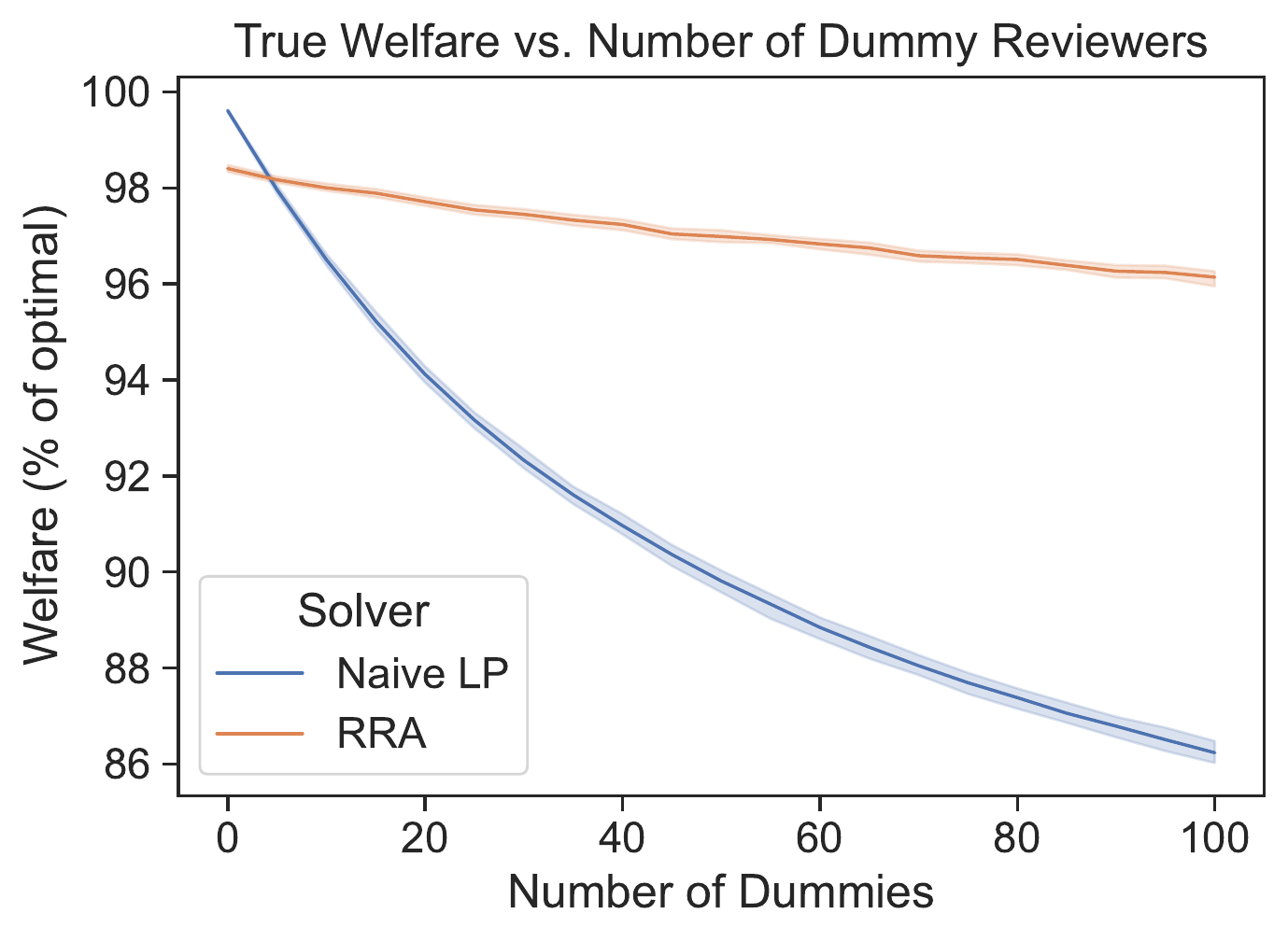}\hfill
    \includegraphics[width=.495\textwidth]{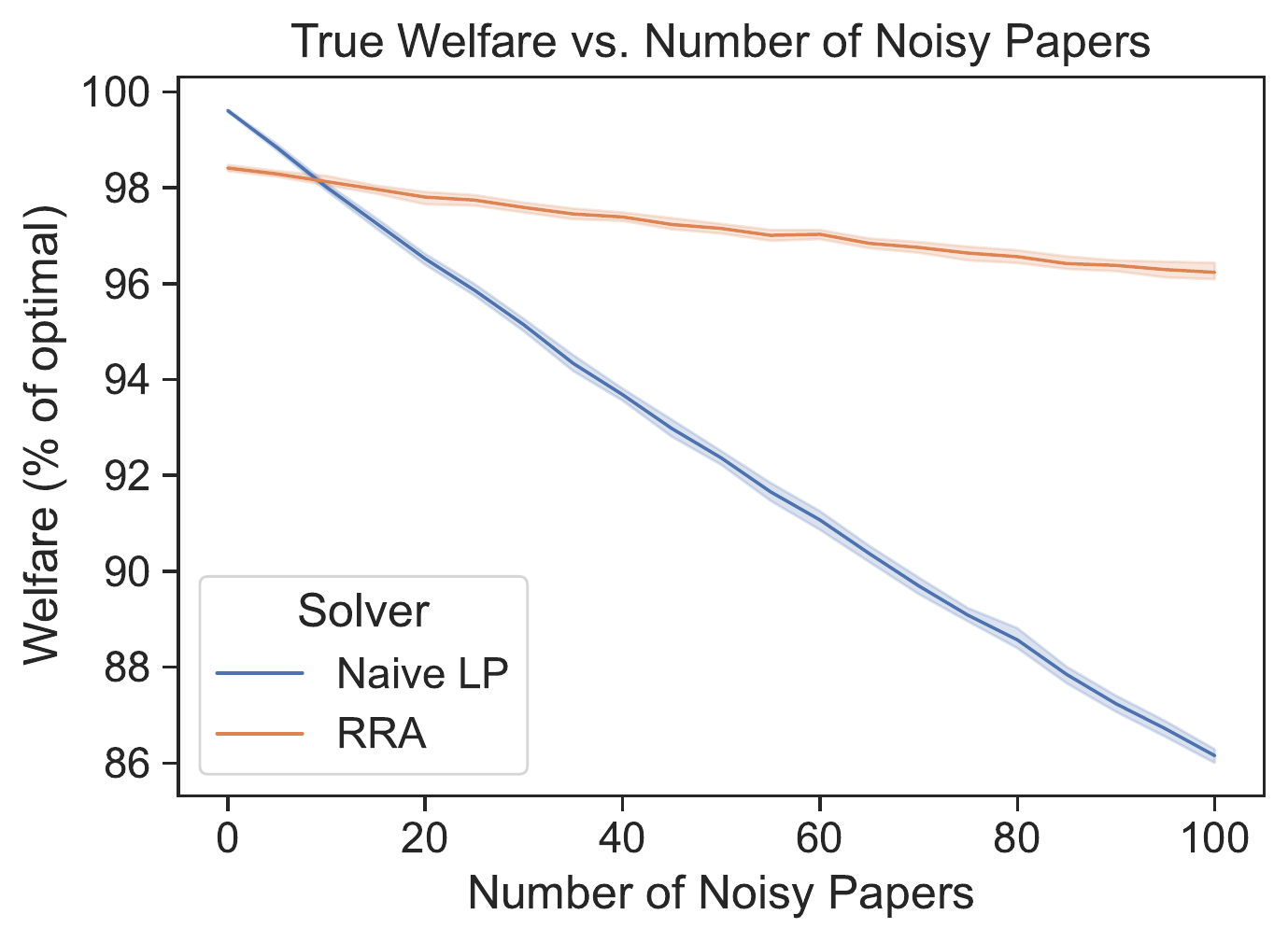}

    \caption{Left: True welfare (as percentage of optimal, when $S^\ast$ is known) of na{\"i}ve \LP{} approach vs. \algoname{} on MIDL 2018 dataset with increasing number of ``dummy'' reviewers. Dummy reviewers have low true affinity for all papers, but higher noise in their estimated affinity scores. The true welfare (which is not known to the conference organizers) decreases sharply for the na{\"i}ve \LP{} formulation as we add more dummy reviewers, but \algoname{} maintains high true welfare. Plot includes min, max, and average for $100$ repetitions at each number of dummy reviewers. Right: True welfare on MIDL 2018 dataset with increasing number of noisy papers. Noisy papers have the $20^{th}$ to $30^{th}$ best reviewers' affinities overestimated by $0.3$. Once again, \algoname{} is robust to this systematic overestimation, while \LP{} degrades poorly. }
        \label{fig:welfarevaryingdummies}
    \vspace*{-8pt}
\end{figure}

To demonstrate the importance of optimizing for the adversarial case, we perform an experiment where we simulate the effect of adding many low-quality, high-variance reviewers to the dataset, and a similar experiment where we systematically overestimate some of the affinities for a subset of papers. 
These reviewers can often appear in modern conference reviewing. 
After being granted limited access to a recent, large AI conference ($m, n \approx 4000$), we manually inspected $100$ random reviewers. $12$ of these reviewers were PhD students, and $40$ of them submitted $10$ or fewer bids. 
PhD students with few papers will tend to have higher variance in true expertise relative to document-based similarity scores derived from their prior work, while reviewers with fewer bids will tend to have higher variance in true interest relative to their bids. 
Intuitively, the systematic overestimation of papers' affinities can occur when paper authors submit a keyword that does not have the exact meaning they expected or is listed by reviewers in the wrong subcommunity. 
We believe that this occurs quite often, and leave a rigorous examination of this important question to future work.

We use the dataset of MIDL 2018, which is publicly available and has been used to validate many reviewer assignment algorithms in the past \cite{kobren2019paper, stelmakh2019peerreview4all, payan2022order}. 
This dataset contains an $n \times m$ matrix of affinity scores that were used to assign reviewers to papers during the conference. 
We use this dataset for this experiment since it is small enough to run a large number of experiments with different settings ($n=118$ papers and $m=177$ reviewers). 
We assume the true affinity of the $177$ original reviewers for each paper is equal to the affinity score present in the public dataset, but is \emph{noisily estimated}. 
Thus, we assume that the conference organizers have access to estimated affinity scores which are equal to the original affinity scores plus normally distributed noise, $\cal N(0, .02)$. 
We also assume the conference organizers know that the estimation error is distributed according to $\cal N(0, .02)$ for each paper-reviewer pair. We also add a number of ``dummy'' reviewers to the dataset. 
For each dummy reviewer, we set the true affinity of that reviewer to be $.1$ for all papers, and sample the estimated affinity from $\cal N(.1, .15)$. 
As with the original reviewers, we assume the conference organizers know the estimation error is distributed according to $\cal N(0, .15)$ for each paper-reviewer pair. 
This simple setup implies a multivariate Gaussian distribution over the true affinity scores, which are unknown to the conference organizers. 
We take a $95\%$ confidence interval of this distribution and intersect it with the unit hypercube to define a truncated-ellipsoidal uncertainty set. 
We then assign reviewers using the na{\"i}ve \LP{} and \algoname{}, and compare the \emph{true}, but unknown, welfare for each approach. 

For the setting with noisy papers, we take a subset of papers and identify the $20^{\textrm{th}}$ to $30^{\textrm{th}}$ (non-inclusive) ranked reviewers for each paper in decreasing order of affinity. 
For each paper, we add $.3$ to the estimated affinity for those $10$ reviewers. 
We then assume that conference organizers estimate the standard deviation of these paper-reviewer pairs to be $.15$, while the remaining standard deviations are estimated at $.02$.

The results are shown in \Cref{fig:welfarevaryingdummies}. 
We report the welfare of na{\"i}ve \LP{} and \algoname{} as a percentage of optimal, averaged over $100$ runs per number of dummy reviewers/papers. We also report the minimum and maximum over all $100$ runs for each setting. As we add more dummy reviewers or noisy papers, the true welfare of the na{\"i}ve \LP{} approach drops by up to roughly $15\%$, while \algoname{} maintains high true welfare.

\section{Extensions to \algoname{}}
\label{sec:extensions}

Our approach is potentially much more general than the cases explored so far in this paper. This approach can incorporate modifications to the objective function, such as incorporation of additional soft constraints. 
We might also consider imposing penalties for unfair outcomes \cite{payan2022order,kobren2019paper,stelmakh2019peerreview4all}, or penalties to discourage manipulative bidding behavior \cite{jecmen2020random, jecmen2022tradeoffs, jecmen2022dataset}. An orthogonal extension is to directly modify the \emph{welfare function}. However, it seems likely that the adversarial nature of our task precludes the use of more sophisticated welfare functions, such as the power-mean welfare, which would otherwise be amenable to maximization \cite{cousins2021axiomatic}.

An orthogonal extension is to directly modify the \emph{welfare function}, which can be used to tune the priorities of the assignment system, for instance more fairness (i.e., group-egalitarian \cite{aziz2023group}, which is the \emph{minimum average utility} of papers in any conference track or subject area) or to consider additional context, such as weighting the importance of conference track, journal track, and workshop track papers with a weighted USW function. We show in Table \ref{tab:iclr} that \algoname{} suffers a reduction in average-case $\USW$ in exchange for greatly improved worst-case $\USW$. It would be fairly straightforward to modify the objective to linearly interpolate between the worst-case and average-case $\USW$ (when the expected value of $S$ is known).

In general, fair welfare functions are concave, and thus the concave optimization methods of this work can be directly applied. However, we need to be able to compute and bound their gradients and modify \algoname{} accordingly, and it may be difficult to solve the adversarial minimization problem for arbitrary welfare functions.
Weighted utilitarian social welfare functions are a straightforward extension that cause essentially no difficulties, and group-egalitarian welfare can be adversarially optimized by separately minimizing for each group, and then taking the score matrix with the smallest group-average (a strategy adopted in the context of fair machine learning under adversarial uncertainty \cite{dong2022decentering}).
However, apart from such convenient special cases, it seems likely that the adversarial nature of our task precludes the use of more sophisticated welfare functions, such as the power-mean welfare, which would otherwise be amenable to maximization \cite{cousins2021axiomatic}.

\section{Conclusion}
\label{sec:conclusion}

Uncertainty is endemic to reviewer assignment.
Bias in bids due to partial information, prediction error in text similarity scores, and keyword terminology misalignment all result in inaccurate estimates of reviewers' abilities. Furthermore, fundamental uncertainty about the quality of future reviews can be quantified, but never fully resolved. We treat uncertainty as a first-class citizen, formulating the problem of assigning reviewers under uncertainty as \problemname{} and addressing it under broad conditions with \algoname{}.

We define and explore the concept of uncertainty sets for reviewer assignment; our theory and examples demonstrate the flexibility of uncertainty sets to model real-world uncertainty-aware reviewer assignment workflows.
Some special cases of uncertainty sets (singleton, hyperrectangular, or spherical) reduce to problems without uncertainty, which can all be solved via linear programming. The general \problemname{} problem is NP-hard, and our \algoname{} algorithm provides approximate solutions with error guarantees dependent on 
the uncertainty set.

Because \algoname{} optimizes against an uncertainty-aware adversary, we avoid assigning paper-reviewer pairs with 
high uncertainty. 
This method of accounting for uncertainty yields 
solutions that are more robust than optimizing with pointwise affinity score estimates. 
We hope this paper serves as a call to further investigate affinity score computation, ensuring affinities correlate with downstream review quality and incorporate the inherent uncertainty present in peer review.
\medskip
\noindent\textbf{Acknowledgments} \quad
    Thanks to OpenReview, Vignesh Viswanathan, Elita Lobo, Chang Zeng, Nihar Shah, and  the reviewers at GAIW, IJCAI, and SAGT for feedback. Cyrus Cousins acknowledges postdoctoral fellowship support from the Center for Data Science at UMass Amherst. This work was performed  using high performance computing equipment obtained under a grant from the Collaborative R\&D Fund managed by the Massachusetts Technology Collaborative.

\bibliographystyle{splncs04}


{
\bibliography{abbshort, main}
}

\clearpage

\appendix

\section{Missing Proofs}

We present proofs for all theoretical results stated in the main paper body without proof. \Cref{appx:sandwich} proves the relation between true welfare and maximin welfare. \Cref{appx:uncsetconstr} proves the results enabling construction of uncertainty sets from estimators as well as by composing or extending other uncertainty sets. Finally, \Cref{appx:solutions} proves NP-hardness of \problemname{} and contains constructive proofs for solving \problemname{} exactly in special cases (hyperrectangles and spheres) as well as approximately in the case of convex uncertainty sets with polynomial-time adversaries.

\subsection{Relation between True Welfare and Maximin Welfare}
\label{appx:sandwich}

We prove the proposition that states that the minimum welfare of an assignment $A$ over a $(\delta, \gamma)$ uncertainty set $\cal S$ relates additively to the true welfare of $A$. The tightness of the approximation depends on the $\cal L_1$ diameter of the set $\cal S$ and the additive $\cal L_1$ error $\gamma$ allowed for the uncertainty set.

\propsandwich*

\begin{proof}
For the right hand side, note that with probability at least $1-\delta$ there exists some $S^\ast_\gamma \in \cal S$ with $\norm{S^\ast_\gamma - S^\ast}_1 \leq \gamma$. Let $S' = \argmin_{S \in \cal S}\USW(A, S)$. By definition, $\USW(A, S') \leq \USW(A, S^\ast_\gamma)$. If $\USW(A, S^\ast_\gamma) \leq \USW(A, S^\ast)$, we have the desired inequality. Otherwise, we can apply the fact that $\USW(A, S^\ast_\gamma) - \USW(A, S^\ast) = \USW(A, S^\ast_\gamma - S^\ast) \leq \frac{1}{n}\norm{S^\ast_\gamma - S^\ast}_1 \leq \frac{1}{n}\gamma$, where the second-to-last inequality holds since every entry of $A$ is in $\{0, 1\}$.

To derive the left hand side, we will aim to bound $\USW(A, S^\ast) - \USW(A, S')$ where $S' = \argmin_{S \in \cal S}\USW(A, S)$. Again with probability at least $1-\delta$ there exists some $S^\ast_\gamma \in \cal S$ with $\norm{S^\ast_\gamma - S^\ast}_1 \leq \gamma$. So $\USW(A, S^\ast) - \USW(A, S^\ast_\gamma) = \USW(A, S^\ast - S^\ast_\gamma) \leq \frac{1}{n}\norm{S^\ast - S^\ast_\gamma}_1 \leq \frac{\gamma}{n}$. Applying similar logic, we can bound $\USW(A, S^\ast_\gamma) - \USW(A, S') = \USW(A, S^\ast_\gamma - S') \leq \frac{1}{n}\norm{S^\ast_\gamma-S'}_1 \leq \frac{L}{n}$.
\end{proof}

\subsection{Constructing Uncertainty Sets}
\label{appx:uncsetconstr}

We first prove our two results showing how ellipsoidal uncertainty sets can be constructed using inductive or transductive tail bounds for estimated affinity scores.

\thminductiveellipsoidaltailbounds*

\begin{proof}


We aim to show the ellipsoid $\cal S$ is a $(\delta, 0)$ uncertainty set. In other words, it must hold that with probability $1-\delta$, the true affinity scores $S^\ast$ satisfy

\[
\frac{1}{nm}(S^\ast \! - \! \hat{S}) \Sigma^{-1} (S^\ast \! - \! \hat{S}) \! \leq \! \smash{\underbrace{\frac{1}{T} \! {\sum_{i=1}^{T}} (f^\ast(p_i, r_i) \! - \! \hat{f}(p_i, r_i))^{2}}_{= \smash{\hat{\xi}}}} \! + \! \smash{\underbrace{\sqrt{\left( \frac{1}{T} \! + \! \frac{n\!+\!m}{nm\alpha_{\min}^{2}} \!\right) \!\frac{\ln \! \smash{\frac{1}{\delta}}}{2}}}_{= \eta}} \vphantom{\sqrt{\frac{1}{T}}} \,  \vphantom{\underbrace{\frac{1}{T}}_{\xi}} \enspace.
\]
We will give a probabilistic bound on the difference
\[
Z \doteq \frac{1}{nm}(S^\ast \! - \! \hat{S})^\intercal \Sigma^{-1} (S^\ast \! - \! \hat{S}) \! -  \smash{\hat{\xi}} = \frac{1}{nm}  \sum_{p=1}^{n} \sum_{r=1}^{m}  \frac{(S_{p,r}^\ast - \hat{S}_{p,r})^{2}}{\alpha(p, r)} - \frac{1}{T} \sum_{i=1}^{T} (f^{\ast}(p_i, r_i) - \hat{f}(p_i, r_i))^{2} \enspace.
\]

We first note that $\Expect[Z] = 0$. We assume that our papers $p \sim \ProbDist^\PapUniv$, our reviewers $r \sim \ProbDist^\RevUniv$, and our samples (which are reviewer-paper pairs) are drawn as $(p_i, r_i) \sim \ProbDist'$. For simplicity of notation, we will write $\{(p_i, r_i)\}_{i=1}^T \distributed \ProbDist'$ to denote the set of all paper-reviewer pairs  $(p_i, r_i) \distributed \ProbDist'$. Also recall that $\alpha(p, r) = \frac{\Prob_{\Pap \distributed \ProbDist^{\PapUniv}} (\Pap = p) \Prob_{\Rev \distributed \ProbDist^{\RevUniv}}(\Rev = r)}{ \Prob_{(\Pap,\Rev) \distributed \ProbDist'\!}((\Pap,\Rev) = (p, r)) }$. We have that
\begin{align*}
\Expect_{\substack{P \distributed \ProbDist^\PapUniv, R \distributed \ProbDist^\RevUniv \\ \{(p_i, r_i)\}_{i=1}^T \distributed \ProbDist'}} [Z] &=
\Expect_{\substack{P \distributed \ProbDist^\PapUniv, R \distributed \ProbDist^\RevUniv \\ \{(p_i, r_i)\}_{i=1}^T \distributed \ProbDist'}} \left[\frac{1}{nm}  \sum_{p=1}^{n} \sum_{r=1}^{m}  \frac{(S_{p,r}^\ast -  \hat{S}_{p,r})^{2}}{\alpha(p, r)} - \frac{1}{T} \sum_{i=1}^{T} (f^{\ast}(p_i, r_i) - \hat{f}(p_i, r_i))^{2}\right] \\ &=
\Expect_{P \distributed \ProbDist^\PapUniv, R \distributed \ProbDist^\RevUniv} \left[\frac{1}{nm}  \sum_{p=1}^{n} \sum_{r=1}^{m}  \frac{(S_{p,r}^\ast - \hat{S}_{p,r})^{2}}{\alpha(p, r)}\right]  - \Expect_{\{(p_i, r_i)\}_{i=1}^T \distributed \ProbDist'}\left[\frac{1}{T} \sum_{i=1}^{T} (f^{\ast}(p_i, r_i) - \hat{f}(p_i, r_i))^{2}\right] \enspace.
\end{align*}
Following the standard argument for importance sampling, it is also clear that 
\begin{align*}
\Expect_{P \sim \ProbDist^\PapUniv, R \sim \ProbDist^\RevUniv} \left[\frac{1}{nm}\sum_{p=1}^{n} \sum_{r=1}^{m}  \frac{(S_{p,r}^\ast - \hat{S}_{p,r})^{2}}{\alpha(p, r)}\right]  &=
\frac{1}{nm} \sum_{p=1}^{n} \sum_{r=1}^{m}  \Expect_{P \sim \ProbDist^\PapUniv, R \sim \ProbDist^\RevUniv} \left[\frac{(S_{p,r}^\ast - \hat{S}_{p,r})^{2}}{\alpha(p, r)}\right] \\ &= \frac{1}{nm}\sum_{i=1}^{nm} \Expect_{(p_i, r_i) \sim \ProbDist'} \left[(f^\ast(p_i,r_i) - \hat{f}(p_i,r_i))^{2}\right] \\
&= \frac{1}{T} \sum_{i=1}^{T} \Expect_{(p_i, r_i) \sim \ProbDist'} \left[(f^\ast(p_i, r_i) - \hat{f}(p_i, r_i))^{2}\right] \\
&= \Expect_{\{(p_i, r_i)\}_{i=1}^T \distributed \ProbDist'}\left[\frac{1}{T} \sum_{i=1}^{T} (f^{\ast}(p_i, r_i) - \hat{f}(p_i, r_i))^{2}\right]
\enspace.
\end{align*}

In addition, $Z$ is a function of a system of negatively-dependent (due to sampling without replacement) random variables comprising $T$ auxiliary terms, $m$ reviewers, and $n$ papers. Modifying any of the $T$ auxiliary terms can result in at most $\frac{1}{T}$ change in $Z$, each of the $n$ papers has impact at most $\frac{1}{n\alpha_{\min}}$, and each of $m$ reviewers has impact at most $\frac{1}{m\alpha_{\min}}$. The sum of square bounded differences is thus $\frac{1}{T} + \frac{1}{n\alpha_{\min}^{2}} + \frac{1}{m\alpha_{\min}^{2}}$.
Consequently, by McDiarmid's inequality
\cite{mcdiarmid1989method}, it holds
\[
\Prob\left( Z \geq \sqrt{\left( \frac{1}{T} + \frac{1}{n\alpha_{\min}^{2}} + \frac{1}{m\alpha_{\min}^{2}} \right) \frac{\ln \frac{1}{\delta}}{2}} \right) \leq \delta \enspace.
\]
%
\end{proof}

\thmtransductiveellipsoidaltailbounds*

\begin{proof}
Proof of result follows similarly to that of \cref{thm:inductiveellipsoidal-tail-bounds}.
Here, reviewers and papers are fixed. We will bound the deviation from the mean for the random variable
\[
Z \doteq \frac{1}{nm} \sum_{p=1}^{n} \sum_{r=1}^{m}  \frac{(S_{p,r} - \hat{S}_{p,r})^{2}}{\alpha(p, r)} - \frac{1}{T} \sum_{i=1}^{T} (f^\ast(p_i, r_i) - \hat{f}(p_i, r_i))^{2} \enspace.
\]

We show that $\Expect[Z] = 0$, as 
\begin{align*}
    \Expect\limits_{\{(p_i, r_i)\}_{i=1}^T \distributed \ProbDist'}[Z] &= \Expect\limits_{\{(p_i, r_i)\}_{i=1}^T \distributed \ProbDist'} \left[\frac{1}{nm} \sum_{p=1}^{n} \sum_{r=1}^{m}  \frac{(S_{p,r} - \hat{S}_{p,r})^{2}}{\alpha(p, r)} - \frac{1}{T} \sum_{i=1}^{T} (f^\ast(p_i, r_i) - \hat{f}(p_i, r_i))^{2} \right] \\
    &= \frac{1}{nm} \sum_{p=1}^{n} \sum_{r=1}^{m}  \frac{(S_{p,r} - \hat{S}_{p,r})^{2}}{\alpha(p, r)} - \frac{1}{T} \sum_{i=1}^{T} \Expect\limits_{\{(p_i, r_i)\}_{i=1}^T \distributed \ProbDist'} \left[(f^\ast(p_i, r_i) - \hat{f}(p_i, r_i))^{2} \right] \\
    &= \frac{1}{nm} \sum_{p=1}^{n} \sum_{r=1}^{m}  \frac{(S_{p,r} - \hat{S}_{p,r})^{2}}{\alpha(p, r)} - \frac{1}{nm} \sum_{i=1}^{nm} \Expect\limits_{\{(p_i, r_i)\}_{i=1}^{nm} \distributed \text{Unif}(P \times R)} \left[\frac{(f^\ast(p_i, r_i) - \hat{f}(p_i, r_i))^{2}}{\alpha(p_i, r_i)}\right] \\ 
    &= \frac{1}{nm} \sum_{p=1}^{n} \sum_{r=1}^{m}  \frac{(S_{p,r} - \hat{S}_{p,r})^{2}}{\alpha(p, r)} - \frac{1}{nm} \sum_{p=1}^{n} \sum_{r=1}^{m}  \frac{(S_{p,r} - \hat{S}_{p,r})^{2}}{\alpha(p, r)}
    \enspace.
\end{align*}

To apply the final bound, we need only consider the $T$ independent auxiliary terms, each of bounded difference $\frac{1}{T}$.
In this case, we have the bounded difference term $\frac{1}{T}$ for each of $T$ auxiliary terms.  
This yields, via the Hoeffding
\cite{hoeffding1963probability} or McDiarmid \cite{mcdiarmid1989method} inequalities, the final result
\[
\Prob\left( Z \geq \sqrt{\frac{\ln \frac{1}{\delta}}{2T}} \right) \leq \delta \enspace.
\]
\cyrus{gamma term for distribution shift? Basically, let's prove the comment at the end of the theorem statement.}
\end{proof}

Now that we have shown how to construct uncertainty sets from affinity score estimators, we finish by showing how to compose uncertainty sets. We proved \Cref{lemma:uncsetint} in \Cref{subsec:compositional}, showing that uncertainty sets can be intersected to form new uncertainty sets. \Cref{lemma:uncsetint} can be combined with \Cref{lemma:l1-error-contract} to create many types of uncertainty sets from simpler uncertainty sets.

\lemerrorcontract*

\begin{proof}


Denote the true affinity scores as $S^\ast$. If $S^\ast \in \cal S'$, then we have $\norm{S - S^\ast}_1 = 0 \leq \gamma - \eta$.
We argue the case when $S^\ast \notin \cal S'$. 
By \Cref{defn:dg_uncert}, with probability at least $1- \delta$, there is some $S \in \cal S$ such that $\norm{S - S^\ast}_1 \leq \gamma$. We can assume that $\eta \leq \norm{S^\ast - S}_1$, since if not, then $S^\ast \in \cal S'$. Consider $S' = S + \eta \frac{S^\ast - S}{\norm{S^\ast - S}_1}$. 
To show $S' \in \cal S'$, we must show that $\left\lVert \eta \frac{S^\ast - S}{\norm{S^\ast - S}_1} \right \rVert_1 \leq \eta$. This can be shown easily, as $\left\lVert \eta \frac{S^\ast - S}{\norm{S^\ast - S}_1} \right \rVert_1 = \frac{\eta}{\norm{S^\ast - S}_1}\norm{S^\ast - S}_1 = \eta$. 
In addition, 
\begin{align*}
    \norm{S^\ast - S'}_1 &= \left\lVert  S^\ast - \left( S + \eta \frac{S^\ast - S}{\norm{S^\ast - S}_1} \right) \right\rVert_1 \\
    &= \left(1-\frac{\eta}{\norm{S^\ast - S}_1}\right) \norm{S^\ast - S}_1 \\
    &= \norm{S^\ast - S}_1 - \eta \\
    &\leq \gamma - \eta \enspace,
\end{align*}
where the final inequality holds with probability at least $1-\delta$. Therefore, with probability at least $1-\delta$, $\cal S'$ contains a point $S'$ within $\gamma - \eta$ $\cal L_1$ distance from $S^\ast$, and $\cal S'$ is a $(\delta, \gamma - \eta)$ uncertainty set.
%
%
%
%
\end{proof}

\subsection{NP-hardness and Algorithms for Solving \problemname{}}
\label{appx:solutions}

We first show that \problemname{} is NP-hard, even for the case of bounded polytopes formed by intersections of polynomially many halfspaces. 

\thmhardness*
\begin{proof}
We reduce from the problem of maximizing egalitarian welfare under the reviewer assignment constraints, which is known to be NP-hard \cite{garg2010assigning,stelmakh2019peerreview4all}. 
In the maximal egalitarian welfare problem for reviewer assignment, we have the same feasible set of assignments $\mathcal{A}$, as well as a fixed score matrix $S$. The goal is to find an assignment $A$ maximizing  the minimum total score of any paper, or equivalently $\min_{p \in P}\sum_{r \in R} A_{p,r}S_{p,r}$.

Given an instance of max egalitarian welfare $(P, R, \mathcal{A}, S)$, we now construct an instance of  \problemname{} over a convex uncertainty set $\mathcal{S}$.
The set of $n$ papers $P$, $m$ reviewers $R$, and valid assignments $\mathcal{A}$ remain the same. 
To construct an uncertainty set $\mathcal{S}$, 
first consider the set $\mathcal{S}' = \{S^{(1)}, \dots, S^{(n)}\}$, where $S^{(p)}$ is defined so that $S^{(p)}_{i,j} = S_{i,j}$ when $i = p$ and $S^{(p)}_{i,j} = 0$ when $i \neq p$. 
Let $\mathcal{S}$ be the convex hull of $\mathcal{S}'$. In other words, $\cal S$ is the set 
\begin{align*}
\mathcal{S} \doteq \left\{ X \in [0, 1]^{n \times m} \,\middle|\, X = \sum_{i =1}^n \alpha_i S^{(i)}, \sum_{i=1}^n \alpha_i = 1 \right\} \enspace.
\end{align*}
Since $\cal S$ is a convex set, $(P, R, \mathcal{A}, \mathcal{S})$ is a \problemname{} instance over a convex uncertainty set. Furthermore, $\cal S$ is a bounded polytope formed by the intersection of polynomially many halfspaces.

For any $A \in \mathcal{A}$, we have that
\begin{align*}
    S_\text{min} \doteq \argmin_{S \in \cal S} n \USW(A, S) = \argmin_{S \in \cal S} \USW(A, S) \enspace. 
\end{align*}
We will show that $S_{\text{min}} \in \cal S'$, completing the proof.
Fix $A$ and consider some element $X \in \mathcal{S}$, where $X = \sum_{i=1}^n \alpha_i S^{(i)}$ and (w.l.o.g.) $\alpha_1, \alpha_2 > 0$. 
Assume also w.l.o.g. that $\sum_{r \in R} A_{1,r}S_{1,r} \leq \sum_{r \in R} A_{2,r}S_{2,r}$. 
Consider the alternative $X' \in \mathcal{S}$ with $X' = (\alpha_1 + \alpha_2) S^{(1)} + \sum_{i=3}^n \alpha_i S^{(i)}$. 
Then
\begin{align*}
    n\USW(A, X) &= \alpha_1 \sum_{r \in R} A_{1,r}S_{1,r} + \alpha_2 \sum_{r \in R} A_{2,r}S_{2,r} + \sum_{i=3}^n \alpha_i \sum_{r \in R} A_{i,r}S_{i,r} \\
    &\geq (\alpha_1 + \alpha_2) \sum_{r \in R} A_{1,r}S_{1,r} + \sum_{i=3}^n \alpha_i \sum_{r \in R} A_{i,r}S_{i,r} \\
    &= n\USW(A, X') \enspace.
\end{align*}

Since we can always decrease $\USW(A, S)$ by restricting the support of the score matrices (i.e., setting some $\alpha_i \neq 0$ to be $0$), the minimal $\USW$ is reached for some $S\in \mathcal{S}'$.
Therefore, since for any given $A$, we have that
\begin{align*}
    \min_{S \in \mathcal{S}'} n\USW(A, S) = \min_{p \in P} \sum_{r \in R} A_{p,r}S_{p,r} \enspace,
\end{align*}
the maximizer for the \problemname{} problem and the maximizer for the egalitarian welfare problem are equivalent. In other words, if we can efficiently compute a robust solution to \problemname{}, we can solve the egalitarian welfare problem. Therefore \problemname{} is NP-hard.
\end{proof}
Although \problemname{} is NP-hard for even fairly simple cases, there are certain special cases in which \problemname{} reduces to a problem solvable in polynomial time. The first case is that of hyperrectangular uncertainty sets.

\thmbox*
\begin{proof}
For any assignment $A \in \mathcal{A}$, $\inf_{S \in \mathcal{S}_{\Box}} \USW(A, S) = \USW(A, \underline{S})$, i.e., the lowest welfare is achieved (possibly non-uniquely) if we assume that all scores are the worst they can be under the rectangular constraints.  
To see why, consider any matrix $X \in \mathcal{S}_{\Box}$ with $X_{p,r} > \underline{S}_{p,r}$ for some $p$ and $r$. 
Consider $X'$ where $X'_{p,r} = \underline{S}_{p,r}$ and $X'_{i,j} = X_{i,j}$ for $i \neq p$ or $j \neq r$. Either $A_{p,r} = 0$ and hence $\USW(A, X) = \USW(A, X')$ or $A_{p,r} = 1$ and thus $\USW(A, X) > \USW(A, X')$. 

We also observe that the problem $\argmax_{A \in \mathcal{A}} \USW(A, \underbar{S})$ is an instance of \problemname{} with known $S$. This problem is solvable in polynomial time via linear programming, as discussed at the beginning of \Cref{sec:uncertainty_set_types}.
\end{proof}

\problemname{} is also solvable in polynomial time for spherical uncertainty sets, and the proof provides a bit more insight than the proof of \Cref{thm:box}.

\thmmaxuswrobust*
\begin{proof}
Let $\mathbf{0}$ denote the $n \times m$ matrix of all $0$ values. For any $A \in \cal A$, it holds that
\begin{align*}
     \inf_{S \in B_{\varepsilon}(S^0)} \USW(A, S) &= \min_{X \in B_{\varepsilon}(\mathbf{0})}  \USW(A, S^0 + X) \\
     &= \min_{X \in B_{\varepsilon}(\mathbf{0})} \USW(A, X) + \USW(A, S^0) & \textsc{Linearity of $\USW$} \\
    &= \USW \left( A, S^0 - A \frac{\varepsilon}{\norm{A}_F} \right) & \textsc{See Below} \\
    &= \USW(A, S^0) - \USW \left( A, A \frac{\varepsilon}{\norm{A}_F} \right) & \textsc{Linearity of $\USW$} \\
    &= \USW(A, S^0) - \mathsmaller{\frac{1}{n}}\norm{A}_F \varepsilon \\
    &= \USW(A, S^0) - \mathsmaller{\frac{1}{n}}\sqrt{K} \varepsilon
    \label{eqn:spherecase} \enspace.
\end{align*}
Because $\USW(A, S^0)$ does not depend on $X$, we just need to find $\argmin_{X \in B_{\varepsilon}(\mathbf{0})} \USW(A, X)$.
Because $\norm{A}_F$ is fixed (and the argmin occurs at $\norm{X}_F = \varepsilon$) the argmin must be 
$X = -A\frac{\varepsilon}{\norm{A}_F}$. 

Again, $\sqrt{K}$, $\varepsilon$, and $n$ are all fixed. So we see that our function has the same argmax as $\USW(A, S^0)$.
The problem $\argmax_{A \in \mathcal{A}} \USW(A, S^0)$ is an instance of \problemname{} with known $S$, which is solvable in polynomial time.

We can quantify the worst-case welfare loss due to spherical uncertainty with radius $\varepsilon$ as 
\begin{align*}
\USW(A, S^0) - \USW(A, S^\ast) \leq \mathsmaller{\frac{1}{n}}\sqrt{K} \varepsilon \enspace.
\end{align*}
A similar argument to the above shows that 
\begin{align*}
     \max_{S \in B_{\varepsilon}(S^0)} \USW(A, S) &= \USW(A, S^0) + \mathsmaller{\frac{1}{n}}\sqrt{K} \varepsilon \enspace,
\end{align*}
and thus we also have
\begin{align*}
\USW(A, S^\ast) - \USW(A, S^0)  \leq \mathsmaller{\frac{1}{n}}\sqrt{K} \varepsilon \enspace.
\end{align*}
\end{proof}

Finally, we will prove \Cref{prop:subgradconv} from \Cref{sec:approx_algo}, which enables approximate solutions to \problemname{} when the uncertainty set $\cal S$ admits a polynomial-time minimizing adversary. We use the following result from \cite{shor1985minimization}.

\begin{restatable}{theorem}{thmsubgradorig}\label{thm:subgradorig}
Suppose that $w: \R^n \to \R$ is a convex function with minimizer $w^\ast$ and optimal set $X^\ast = \{\vec x \in \R^n : w(\vec x) = w^\ast\}$. Let $\vec x_0$ denote the starting point for subgradient descent. Suppose we have an upper bound $\rho \geq \max_{\vec x \in X^\ast} \norm{\vec x_0 - \vec x}_2$. Suppose we also have an upper bound $\lambda$ for the Euclidean norm of the subgradient at each step of the subgradent descent algorithm. After $T$ steps of subgradient descent with constant step size $\alpha$, if $\hat{w}$ represents the best function value found after $T$ steps, then
\begin{align*}
\hat{w} - w^\ast \leq \frac{\rho^2 + \lambda^2 \alpha^2 T}{2 \alpha T} \enspace.
\end{align*}
\end{restatable}

\propsubgradconv*

\begin{proof}

First, we show that $\rho \leq \sqrt{2K}$ (recall $\rho$ is defined as the upper bound $\rho \geq \max_{\vec x \in X^\ast} \norm{\vec x_0 - \vec x}_2$). Because the initial assignment $A^{(0)}$ and the optimal assignment $A^\ast$ both lie in $\tilde{\mathcal{A}}$, they both must satisfy the constraint that for all papers $p$, $\sum_{r \in R} A_{p,r} = k_p$. This implies that even if there is no overlap in the two assignments, they will differ on at most $2K$ entries. Taking the Euclidean norm over the difference of the two assignments in the worst case gives the bound on $\rho$.

Our function is concave, and we are maximizing over the input space, so we can apply \Cref{thm:subgradorig} to minimize its negative. From \Cref{thm:subgradorig} we know that after $T$ iterations of subgradient ascent with step size $\alpha$, the error $\varepsilon = \inf_{S \in \mathcal{S}}\USW(A^\ast, S) - \inf_{S \in \mathcal{S}}\USW(\hat{A}, S)$ satisfies
\[
\varepsilon \leq \frac{\rho^{2} + \lambda^{2} \alpha^{2} T}{2\alpha T} \enspace.
\]

The right-hand side is minimized at
$ \begin{aligned} \alpha = \frac{\rho}{\lambda\sqrt{T}} \end{aligned}$.
%
%
If we substitute this value of $\alpha$, we get
\[
\varepsilon \leq \frac{\rho \lambda}{\sqrt{T}} \enspace.
\]
Solving for $T$ then yields
\[
T \geq \Big(\frac{\rho \lambda}{\varepsilon}\Big)^{2} \enspace.
\]
Since $\rho \geq \sqrt{2K}$, the larger number of iterations $T \geq 2K( \frac{\lambda}{\varepsilon})^2$ will suffice. Substituting this value of $T$ into the equation given for $\alpha$ above yields $\alpha = \frac{\varepsilon}{\lambda^2}$.
\end{proof}

We also prove the following corollary, which gives a more precise bound on the number of required iterations when $\cal S$ is a (truncated) ellipsoid.

\corellipsoidalconvergence*

\begin{proof}

Because the subgradient is always contained in $\mathcal{S}$, we can easily upper-bound the norm of it. We have for all $A \in \tilde{\cal A}$ and $S \in \cal S$ that
\begin{align*}
    \norm{\nabla_A \USW(A, S)}_2 \leq \max_{S \in \mathcal{S}} \mathsmaller{\frac{1}{n}} \norm{S}_2 \enspace.
\end{align*}

In turn, if the ellipsoid has center $\vec \mu_\cal S$ and radius $q$, then we can use the triangle inequality to show
$$\max_{S \in \mathcal{S}} \norm{S}_2 \leq \norm{\vec \mu_{\mathcal{S}}}_2 + q \enspace,$$
where $q$ can
be computed as the maximum eigenvalue of $\Sigma_\mathcal{S}$. Both $\norm{\vec \mu_{\mathcal{S}}}_2$ and $q$ are $O(\sqrt{nm})$, so we see that $\lambda$ is $O(\sqrt{\frac{m}{n}})$. Applying \Cref{prop:subgradconv} completes the proof.
\end{proof}

\thmgaps*
\begin{proof}
We begin with the maximin gap bound.
%
With probability at least $1-\delta$, there is some $S^\ast_\gamma \in \cal S$ such that $\norm{S^\ast_\gamma - S^\ast}_1 \leq \gamma$, and thus $\USW(A^\ast, S^\ast) \leq \USW(A^\ast, S^\ast_\gamma) + \frac{\gamma}{n}$ (by \Cref{prop:sandwich}). Likewise, $\USW(A^\varepsilon, S^\ast_\gamma) \leq \USW(A^\varepsilon, S^\ast) + \frac{\gamma}{n}$. 
Together, these imply
\begin{align*}
\USW(A^\ast, S^\ast) - \USW(A^\varepsilon, S^\ast) &\leq \USW(A^\ast, S^\ast_\gamma) - \USW(A^\varepsilon, S^\ast) + \mathsmaller{\frac{\gamma}{n}} \\
&\leq \USW(A^\ast, S^\ast_\gamma) - \USW(A^\varepsilon, S^\ast_\gamma) + 2\mathsmaller{\frac{\gamma}{n}} \enspace.
\end{align*}
Now, $\begin{aligned}\USW(A^\ast, S^\ast_\gamma) \leq \max_{A \in \cal A}\inf_{S \in \cal S} \USW(A, S) + \mathsmaller{\frac{L}{n}}\end{aligned}$, and by definition $\begin{aligned} \USW(A^\varepsilon, S^\ast_\gamma) \geq \inf_{S \in \cal S}\USW(A^\varepsilon, S) \end{aligned}$, thus 
\begin{align*}
\USW(A^\ast, S^\ast_\gamma) - \USW(A^\varepsilon, S^\ast_\gamma) + 2\frac{\gamma}{n} \leq \max_{A \in \cal A}\inf_{S \in \cal S} \USW(A, S) - \inf_{S \in \cal S}\USW(A^\varepsilon, S) + \mathsmaller{\frac{2\gamma + L}{n}} \enspace.
\end{align*}
The definition of $A^\varepsilon$ implies $\begin{aligned}\max_{A \in \cal A}\inf_{S \in \cal S} \USW(A, S) - \inf_{S \in \cal S}\USW(A^\varepsilon, S) \leq \varepsilon\end{aligned}$, yielding the desired result.

To obtain the expected regret of \algoname{}, we first apply the two facts that $\Expect_{A'}[A'] = \tilde{A}^\varepsilon$ and $\USW$ is a linear objective function. We now must bound $\USW(A^\ast, S^\ast) - \USW(\tilde{A}^\varepsilon,S^\ast)$, which we can do using the same proof we used to obtain the maximin gap above (replacing $A^\varepsilon$ and $\cal A$ with $\tilde{A}^\varepsilon$ and $\tilde{\cal A}$).

Finally, the probabilistic worst-case regret bound follows from Markov's inequality. 
\end{proof}

\end{document}